\pdfoutput=1
\RequirePackage{amsmath}
\documentclass{llncs}
\usepackage{amsmath,amssymb}
\usepackage[T1]{fontenc}
\usepackage{graphicx}
\usepackage{lineno}
\usepackage{hyperref}
\usepackage{authblk}
\usepackage{color}
    
\usepackage[dvipsnames, table]{xcolor}
\usepackage{wrapfig}
\usepackage[capitalise]{cleveref}
\usepackage{hyperref}
\usepackage{thmtools}
\usepackage{thm-restate}
\usepackage{enumitem}
\usepackage{mathrsfs}  
\usepackage[all]{xy}
\usepackage{multirow}
\usepackage{caption} 
\usepackage{float}
\usepackage{multicol}
\usepackage{soul}
\usepackage{hhline}
\usepackage{bbm}
\usepackage{mathtools}
\usepackage{bm}
\usepackage{ dsfont }
\usepackage{pifont}
\usepackage{array}
\usepackage{pdflscape}
\usepackage{mdframed}
\usepackage{tikz-cd}
\usepackage{tikz}
\usetikzlibrary{patterns.meta}
\usepackage{pgfplots}
    \pgfplotsset{compat=newest}
\usepackage[ttdefault=true]{AnonymousPro}
\usepackage{adjustbox}
\usepackage{makecell}
\usepackage{tabstackengine}
\usepackage{xspace}
\usepackage{soul}
\usepackage[pass]{geometry}
\usepackage{wrapfig}
\usepackage[misc]{ifsym}
\usepackage{orcidlink}
\usepackage{scalerel}
\usepackage{xcolor}
    \hypersetup{
        colorlinks,
        linkcolor={green!50!black},
        citecolor={red!85!black},
        urlcolor={blue!85!black}
    }
\newcommand{\s}{\ensuremath{\sigma}}

\renewcommand{\P}{\ensuremath{\mathcal{P}}}

\newcommand{\ax}{\textit{Ax}}

\newcommand{\wit}{\textit{wit}}

\newcommand{\overarrow}{\overrightarrow}

\newcommand{\wrt}{\textit{w.r.t.} }

\newcommand{\Exists}[1]{\exists\,#1.\:}

\newcommand{\Forall}[1]{\forall\,#1.\:}

\newcommand{\GT}[1]{{}}

\newcommand{\A}{\ensuremath{\mathcal{A}}}

\newcommand{\B}{\ensuremath{\mathcal{B}}}

\newcommand{\C}{\ensuremath{\mathcal{C}}}

\newcommand{\F}{\ensuremath{\mathcal{F}}}

\newcommand{\T}{\ensuremath{\mathcal{T}}}

\renewcommand{\P}{\ensuremath{\mathcal{P}}}

\newcommand{\vars}{\textit{vars}}

\newcommand{\N}{\ensuremath{\mathbb{N}_{\omega}}}

\newcommand{\eq}[1]{\ensuremath{Eq(#1)}}

\newcommand{\minmod}{\ensuremath{\textbf{minmod}}}

\newcommand{\Teven}{\ensuremath{\T_{even}^{\infty}}}

\newcommand{\cvciii}{\ensuremath{\textsf{CVC3}}}
\newcommand{\cvcfour}{\ensuremath{\textsf{CVC4}}}
\newcommand{\cvcfive}{\ensuremath{\textsf{cvc5}}}

\newcommand{\Tinfty}{\ensuremath{\T_{\infty}}}

\newcommand{\TM}{\ensuremath{\T_{f}}}

\renewcommand{\P}{\ensuremath{\mathcal{P}}}

\newcommand{\NNEQ}[1]{\ensuremath{\neq(#1_{1},\ldots,#1_{n})}}

\newcommand{\dif}[1]{\ensuremath{\psi^{orb}_{\geq #1}}}

\newcommand{\orb}[1]{\ensuremath{\psi^{orb}_{=#1}}}

\newcommand{\eqs}{{\scaleobj{0.7}{=}}}

\newcommand{\difs}{\scaleobj{0.7}{\neq}}

\newcommand{\distinct}[1]{\NNEQ{x}}

\renewcommand{\int}[2]{\mathcal{#1}/\mathcal{#2}}

\newcommand{\Tp}{\ensuremath{\T_{P}^{\nf}}}

\newcommand{\Tot}{\ensuremath{\T_{\textit{orb}}^{2}}}

\newcommand{\Teq}{\ensuremath{\T_{=}}}

\newcommand{\Tgr}{\ensuremath{\T_{\leq}}}

\newcommand{\Tgp}{\ensuremath{\T_{\leq}^{orb}}}

\newcommand{\Spn}{\ensuremath{\Sigma_{P}^{n}}}

\newcommand{\phieqs}{\ensuremath{\varphi^{\eqs}_{\geq n}}}

\newcommand{\Tb}{\ensuremath{\T_{\infty}^{\nf}}}

\newcommand{\No}{\ensuremath{\mathbb{N}^{*}}}

\newcommand{\Sone}{\ensuremath{\Sigma_{1}}}

\newcommand{\Inf}{\ensuremath{\aleph_{0}}}

\newcommand{\eqp}[1]{\ensuremath{Eq^{\prime}(#1)}}

\newcommand{\eqpp}[1]{\ensuremath{Eq^{\prime\prime}(#1)}}

\newcommand{\spec}{\ensuremath{\textit{Spec}}}

\newcommand{\dom}[1]{\ensuremath{dom(#1)}}

\newcommand{\ns}{\ensuremath{t}}

\newcommand{\os}{\ensuremath{s}}

\newcommand{\nf}{\ensuremath{h}}

\newcommand{\nft}{\ensuremath{g}}

\newcommand{\of}{\ensuremath{f}}

\newcommand{\Ss}{\ensuremath{\Sigma_{s}}}

\newcommand{\Sat}{\ensuremath{\Sigma_{t}^{a}}}

\newcommand{\TMn}{\ensuremath{\T_{\nft}}}

\newcommand{\NFT}{\ensuremath{G}}

\tikzset{
    rotated halfcircle/.style={%
        mark=halfcircle*,
        mark color=black,
        fill=red,
        every mark/.append style={rotate=#1}
    }
}

\newcolumntype{P}[1]{>{\centering\arraybackslash}p{#1}}

\Crefname{theorem}{Theorem}{Theorems}
\Crefname{lemma}{Lemma}{Lemmas}
\Crefname{corollary}{Corollary}{Corollaries}
\Crefname{example}{Example}{Examples}
\Crefname{proposition}{Proposition}{Propositions}
\Crefname{definition}{Definition}{Definitions}
\Crefname{conjecture}{Conjecture}{Conjectures}

\let\tp\texorpdfstring

\usepackage{etoolbox}
\AtEndEnvironment{proof}{\phantom{}\qed} 

\begin{document}
\title{Being polite is not enough \\ (and other limits of theory combination)}
%
%
\author{Guilherme V. Toledo\inst{1} 
\and
Benjamin Przybocki \inst{2}
\and
Yoni Zohar \inst{1}
}
%
%


\institute{Bar-Ilan University, Israel
\and 
University of Cambridge, UK
}
\maketitle              

\setcounter{page}{1} 
\pagestyle{plain} 

\begin{abstract}
In the Nelson--Oppen combination method for satisfiability modulo theories, the combined theories must be stably infinite;
in gentle combination, one theory has to be gentle, and the other has to satisfy a similar yet weaker property;
in shiny combination, only one
has to be shiny (smooth, with a computable minimal model function and the
finite model property);
and for polite combination, only one has to be
strongly polite (smooth and strongly finitely witnessable).
For each combination method, we prove that if any of its assumptions are removed, then there is no general method to combine an arbitrary pair of theories satisfying the remaining assumptions.
We also prove new theory combination results that weaken the assumptions of gentle and shiny combination.\footnote{This work was funded by NSF-BSF grant 2020704, ISF grant 619/21, and
the Colman-Soref fellowship.
}

\end{abstract}
\section{Introduction}\label{Introduction}
Let us start at the middle.
Polite theory combination~\cite{RanRinZar} was not the first 
method 
to combine two theories (see, e.g., \cite{NelsonOppen,DBLP:journals/jar/TinelliZ05}). 
It was also not the last (see, e.g., \cite{gentle,flexible,DBLP:conf/birthday/BonacinaFRT19}).
However, it is one of the most influential approaches 
to theory combination. 
In fact, it 
 has found its way to the implementation of the state-of-the-art SMT-solver
$\cvcfive$~\cite{cvc5} (and also $\cvcfour$~\cite{DBLP:conf/cav/BarrettCDHJKRT11} and
$\cvciii$~\cite{DBLP:conf/cav/BarrettT07}).

The history of polite combination is illustrated in \Cref{venn-all},
that focuses on decidable theories (marked by the large rectangle).
The left circle corresponds to decidable theories that can be combined 
with any other decidable theory over a disjoint signature.
We call such theories {\em combinable}.
In~\cite{RanRinZar}, it was argued that a sufficient condition for combinability is {\em politeness},
a technical notion that concerns cardinalities of models.
%
%
In other words, \cite{RanRinZar} claimed that the red-hatched region of \Cref{venn-all} is  empty.
Then, the paper~\cite{JB10-LPAR} discovered a bug in the proof from~\cite{RanRinZar},
and offered to replace politeness by a seemingly stronger notion, {\em strong politeness} (the name is due to \cite{CasalRasga2}).
It was proved in \cite{JB10-LPAR} that strongly polite theories are combinable, 
which positions the small circle that represents strongly polite theories completely within the circle representing combinable theories.

While \cite{JB10-LPAR} found a bug in the {\em proof} of \cite{RanRinZar},
it left two questions open:
$(i)$~does politeness imply combinability, as \cite{RanRinZar} claimed, only with a different proof? 
and
$(ii)$~do polite theories that are not strongly polite exist?
In terms of \Cref{venn-all}:
$(i)$~is the red-hatched region empty?
and
$(ii)$~is the entire hatched region empty?
Question $(ii)$ was recently resolved
in~\cite{SZRRBT-21,CADE}:
a theory named $\TM$ was found wandering around the hatched region.
However, question $(i)$ remained unanswered, as it was unclear whether $\TM$ resided in the red hatched region
or the blue one. 

In this paper we solve question $(i)$,
by placing
$\TM$ in the red-hatched region:
it is polite but uncombinable.
To show this, we introduce a new decidable theory, named $\Teq$, over a disjoint signature,
and prove that its combination with $\TM$
is {\em undecidable}.

$\TM$ and $\Teq$
are not merely mustard watches~\cite{ringard1990mustard}.
They show that
being polite is not enough (for theory combination),
finally closing the question of politeness vs. combinability, that remained open since \cite{JB10-LPAR}.
Foundationally, they show that the fix of \cite{JB10-LPAR} was indeed necessary.
And practically, they justify the implementation overhead of adopting the more complicated definition of strong politeness.


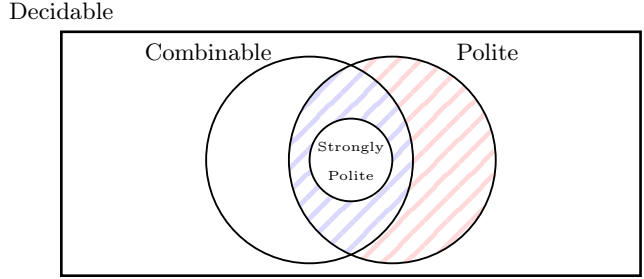
\begin{figure}[t]
\vspace{-4mm}
\centering
    \begin{tikzpicture}[scale=0.55]
\def\firstcircle{(0,0) coordinate (a) circle (1.0cm)}
\def\thirdcircle{(-1,0) coordinate (c) circle (2.5cm)}
\def\secondcircle{(1,0) coordinate (d)  circle (2.5cm)}
\begin{scope}
\fill[
       pattern={Lines[
                  distance=2.2mm,
                  angle=45,
                  line width=0.6mm
                 ]},
        pattern color=red!15
       ] \secondcircle;
    \end{scope}
\begin{scope}
\clip \thirdcircle;
\fill[
       pattern={Lines[
                  distance=2.2mm,
                  angle=45,
                  line width=0.6mm
                 ]},
        pattern color=blue!15
       ] \secondcircle;
    \fill[white] \firstcircle;
    \end{scope}
\draw[line width=0.25mm] \firstcircle;
\draw[line width=0.25mm] \secondcircle;
\draw[line width=0.25mm] \thirdcircle;
\node[label={Combinable}] (B) at (-3.45,2) {};
\node[label={Polite}] (B) at (3.3,2) {};
\node[label={[align=center]\tiny Strongly\\\tiny Polite}] (B) at (0.0,-0.9) {};
\node[label={Decidable}] (B) at (-7,3) {};

 \draw[line width=1pt, black, fill=black, fill opacity=0.0] 
        (-7,-2.8) rectangle (7,3.1);

    \end{tikzpicture}
    \caption{An illustration of the contributions of
\cite{RanRinZar,JB10-LPAR,CADE}
and the current paper.
    }
    \label{venn-all}
    \vspace{-6mm}
\end{figure}

The existence of $\TM$ and $\Teq$ can be seen as a {\em limitation theorem}:
they show that the polite combination method
cannot be applied if strong politeness is weakened to politeness.
We present similar limitation theorems
for other combination methods.
For the Nelson--Oppen method~\cite{NelsonOppen}, we show that if only one of the theories is assumed to be stably infinite, then the combination method fails;
this was previously proven in \cite{Bonacina}, 
but that study did not broach any combination methods other than Nelson--Oppen.
We provide a similar treatment for gentle combination~\cite{gentle}. 
Finally, for a theory to be shiny~\cite{shiny}, it has to satisfy three requirements.
For each one, we show its necessity as well.
Surprisingly, we can reuse $\TM$ and $\Teq$ for almost all limitation theorems, except for one.
Therefore, we use two more theories.
The first, $\Tinfty$, is also taken from \cite{CADE}.
The second, $\Tgr$, is new, and can be seen as a generalization of the theory from \cite{Bonacina} that was used for the Nelson-Oppen limitation theorem.
For all theories (old and new), we prove that they satisfy the required properties for each limitation theorem.
%

Additionally, we prove
two new combination theorems, 
based on the gentle and shiny combination methods.
These theorems relax some of the requirements for theory combination, while  ensuring the decidability of the combined theories.
In a sense, the new theorems remedy the limitation theorems that we prove.

This paper is organized as follows.
\Cref{sec:preliminaries} surveys notions regarding theory combination. 
\Cref{sec:infinitesigs} proves our main theorems, exhibiting limits of common  combination methods.
\Cref{sec:finitesig} improves the proofs of \Cref{sec:infinitesigs} by only using theories over finite signatures.
\Cref{sec:newtheorems} proves new combination theorems.
\Cref{sec:conclusion} concludes and provides directions for future research.


\section{Preliminaries}
\label{sec:preliminaries}

We use $\mathbb{N}$ to denote the set of naturals including $0$, and $\No$ to denote  $\mathbb{N}\setminus\{0\}$.
If $X$ is a set, $|X|$ is its cardinality, and $|\mathbb{N}|=\aleph_{0}$.

\subsection{First-order logic}
\label{sec:FOL}
A {\em signature} is a pair $\Sigma=(\F_{\Sigma},\P_{\Sigma})$ where:
$\F_{\Sigma}$ is a set of function symbols, each with arity $n\in\mathbb{N}$;
and $\P_{\Sigma}$ is a set of predicate symbols, each with arity $m\in\mathbb{N}$, containing at least the equality $=$, of arity $2$.
$\Sigma$ is said to be {\em empty} if it has no function and predicate symbols other than $=$.
Two signatures are said to be {\em disjoint} if the only symbol they share is $=$.
We define {\em terms}, {\em literals}, {\em clauses} (i.e., disjunctions of literals), {\em cubes} (i.e., conjunctions of literals), {\em formulas} and {\em sentences} 
in the usual way.
If $\ns$ is a unary function symbol and $x$ a variable, we define by induction the terms $\ns^{0}(x):=x$ and $\ns^{n+1}(x):=\ns(\ns^{n}(x))$.
The set of variables in a formula $\varphi$ is denoted by $\vars(\varphi)$.

{\em $\Sigma$-interpretations} $\A$ are defined as usual:
$\dom{\A}$ is the {\em domain} of $\A$;
for a function symbol $f$ of arity $n$, $f^{\A}:\dom{\A}^{n}\to\dom{\A}$; 
for a predicate symbol $P$ of arity $m$, $P^{\A}\subseteq\dom{\A}^{m}$;
for a variable $x$, $x^{\A}\in\dom{\A}$.
For a term $\alpha$, $\alpha^{\A}$ is its value in  $\A$, and for a set $\Gamma$  of terms, $\Gamma^{\A}=\{\alpha^{\A} : \alpha\in\Gamma\}$.
If $\A$ satisfies $\varphi$, we write $\A\vDash\varphi$.
%
%
Formulas from \Cref{card-formulas} are satisfied by $\A$ 
when
$|\dom{\A}|$ is: at least $n$ (for $\neq(x_{1},\ldots,x_{n})$ and $\psi_{\geq n}$);
at most $n$ (for $\psi_{\leq n}$);
exactly $n$ (for $\psi_{=n}$).

\begin{figure}[t]
\begin{mdframed}
\vspace{-4mm}
\begin{equation*}
\begin{aligned}
    \neq(x_{1},\ldots,x_{n})=\bigwedge_{i=1}^{n-1}\bigwedge_{j=i+1}^{n}\neg(x_{i}=x_{j})\\
    \\\vspace{-4mm}
    \psi_{\geq n}=\Exists{{x_1\ldots x_n}}\NNEQ{x}
\end{aligned}
\quad
\begin{aligned}
\psi_{\leq n}=\Exists{x_1,\ldots,x_n}\Forall{y}\bigvee_{i=1}^{n}y=x_{i}\\
\\\vspace{-4mm}
\psi_{=n}=\psi_{\geq n}\wedge\psi_{\leq n}
\end{aligned}
\end{equation*}
\end{mdframed}\vspace{-4mm}
\caption{Cardinality formulas.}
\label{card-formulas}
\end{figure}

A {\em theory} $\T$ is a class of all interpretations that satisfy a (finite or infinite) set of sentences $\ax(\T)$ called the {\em axiomatization} of $\T$;
$\varphi$ is said to be {\em $\T$-satisfiable} when there is a $\T$-interpretation satisfying $\varphi$;
it is {\em $\T$-valid} when every $\T$-interpretation
satisfies $\varphi$.
Two formulas $\varphi$ and $\psi$ are {\em $\T$-equivalent} when a $\T$-interpretation satisfies $\varphi$ if and only if it satisfies $\psi$.
$\T$ is {\em decidable} if the set of $\T$-satisfiable quantifier-free formulas is decidable.

%
%
%
%

\subsection{Theory combination theorems}
\label{sec:thcombSMT}
In what follows, $\Sigma$, $\Sigma_1,\Sigma_2$ are 
signatures,
and $\T$, $\T_{1}$, $\T_{2}$ are $\Sigma$, $\Sigma_1,\Sigma_2$-theories, respectively.
We assume $\Sigma_1$ and $\Sigma_2$ are disjoint,
and $\T_1$ and $\T_2$ are decidable.
$\Sigma_1\cup\Sigma_2$ is the signature obtained by collecting all function and predicates symbols from $\Sigma_1$ and $\Sigma_2$. 
$\T_1\oplus\T_2$ is the $\Sigma_1\cup\Sigma_2$-theory axiomatized by $\ax(\T_1)\cup\ax(\T_2)$.

%
We start with Nelson--Oppen.
$\T$ is {\em stably infinite} if for every quantifier-free $\T$-satisfiable formula $\varphi$ there is an infinite $\T$-interpretation $\A$ with $\A\models\varphi$. 

\begin{theorem}[\cite{NelsonOppen}]
\label{originalNO}
$\T_{1}\oplus \T_{2}$ is decidable, if both $\T_1$ and $\T_2$ are stably infinite.
\end{theorem}

Next, we define politeness.
$\T$ is {\em smooth} if for every quantifier-free $\T$-satisfiable formula $\varphi$, $\T$-interpretation $\A$ that satisfies $\varphi$, and cardinal $\kappa>|\dom{\A}|$, there is a $\T$-interpretation $\B$ that satisfies $\varphi$ with $|\dom{\A}|=\kappa$ (notice smoothness implies stable infiniteness, as we can choose an infinite $\kappa$).
%
$\T$ is {\em finitely witnessable} if there exists a function $\wit$ (called a witness) from the quantifier-free formulas of $\Sigma$ into themselves such that, for every quantifier-free formula $\varphi$, one has that: 
$(I)$ $\varphi$ and $\Exists{\overarrow{x}}\wit(\varphi)$ are $\T$-equivalent, where $\overarrow{x}=\vars(\wit(\varphi))\setminus\vars(\varphi)$;
$(II)$ if $\wit(\varphi)$ is $\T$-satisfiable there exists a $\T$-interpretation $\A$ that satisfies $\wit(\varphi)$ with $\dom{\A}=\vars(\wit(\varphi))^{\A}$.
$\T$ is {\em polite} if it is both smooth and finitely witnessable. The following was stated as a theorem in~\cite{ranise:inria-00000570}, but its proof was later refuted in~\cite{JB10-LPAR}. 
It therefore it remained a conjecture, which essentially states that politeness is enough for theory combination.

\begin{conjecture}[\cite{ranise:inria-00000570}]
\label{originalpolite}
$\T_{1}\oplus \T_{2}$ is decidable, provided that $\T_2$ is polite. 
\end{conjecture}

Next: strong politeness. 
Given a finite set of variables $V$ and an equivalence relation $E$ on $V$, the {\em arrangement} induced by $E$ on $V$, denoted by $\delta_{V}^{E}$ or $\delta_{V}$ if $E$ is clear, is the conjunction, for $x,y\in V$, of all formulas $x=y$, if $x E y$, or $\neg(x=y)$ otherwise.
$\T$ is {\em strongly finitely witnessable} if it is finitely witnessable, with witness $\wit$, which in addition satisfies:
$(II^{\prime})$ for every quantifier-free formula $\varphi$, finite set of variables $V$, and arrangement $\delta_{V}$ on $V$, if $\wit(\varphi)\wedge\delta_{V}$ is $\T$-satisfiable then there exists a $\T$-interpretation $\A$ that satisfies $\wit(\varphi)\wedge\delta_{V}$ with $\dom{\A}=\vars(\wit(\varphi)\wedge\delta_{V})^{\A}$.
In that case $\wit$ is called a strong witness.
$\T$ is {\em strongly polite} if it is smooth and strongly finitely witnessable (it was shown in~\cite[Theorem~2]{nounicorns} that in this definition, smoothness can be replaced by stable infiniteness).

\begin{theorem}[\cite{JB10-LPAR}]
\label{originalstrongpolite}
$\T_{1}\oplus \T_{2}$ is decidable, provided that $\T_2$ is strongly polite.
\end{theorem}

We continue to shiny combination.
$\T$ has the {\em finite model property} if, for every quantifier-free $\T$-satisfiable formula 
$\varphi$%
\footnote{This notion is often not restricted to a quantifier-free $\varphi$, but in SMT it usually is.} and $\T$-interpretation $\A$ that satisfies $\varphi$, there exists a $\T$-interpretation $\B$ that satisfies $\varphi$ with $|\dom{\B}|<\aleph_{0}$.
%
Let $\N=\mathbb{N}\cup\{\aleph_{0}\}$.
The {\em minimal model function} $\minmod$ of $\T$ is a function from the quantifier-free formulas of $\Sigma$ to $\N$ such that for every quantifier-free $\T$-satisfiable
formula $\varphi$:
$(I)$ there exists a $\T$-interpretation $\A$ that satisfies $\varphi$ with $|\dom{\A}|=\minmod(\varphi)$;
$(II)$ if $\B$ is a $\T$-interpretation that satisfies $\varphi$, $\minmod(\varphi)\leq |\dom{\B}|$.\footnote{The function $\minmod$ is only guaranteed to exist if $\F_{\Sigma}\cup\P_{\Sigma}$ is countable.}
%
$\T$ is {\em shiny} if it is smooth, and it has both the finite model property and a computable minimal model function.
Note that \cite{CasalRasga,CasalRasga2} showed that shininess is equivalent 
to strong politeness for decidable theories.
In our context, however, we disassemble these notions to their more rudimentary ingredients,
and when doing so,
the equivalence does not necessarily hold.

\begin{theorem}[\cite{DBLP:journals/jar/TinelliZ05}]
\label{originalshiny}
$\T_{1}\oplus \T_{2}$ is decidable, provided that $\T_2$ is shiny.
\end{theorem}

 The {\em spectrum} $\spec(\T,\varphi)$ of $\T$ w.r.t. a quantifier-free formula $\varphi$ is the set of countable cardinalities of $\T$-interpretations that satisfy $\varphi$.
Roughly,
$\T$ is {gentle} if, given a conjunction $\varphi$ of literals, $\spec(\T,\varphi)$ can be computed, and is either a finite set of finite cardinalities or a co-finite%
\footnote{A set $S$ is co-finite if its complement $\mathbb{N}\setminus S$ is finite.} 
set of cardinalities. Formally, $\T$ is \emph{gentle} if there is an algorithm that, for every conjunction $\varphi$ of literals, outputs a pair $(b,S)$, with $b$ a boolean and $S \subset \No$ is finite, such that $(i)$~if $b$ is true, then $\spec(\T,\varphi) = S$ and $(ii)$~if $b$ is false, then $\spec(\T,\varphi) = \N \setminus S$. 
Note that a gentle theory is decidable, because $\varphi$ is $\T$-satisfiable if and only if $\spec(\T,\varphi) \neq \emptyset$.

\begin{theorem}[{\cite[Theorem~3]{gentle}}] \label{gentle-recovery-fontaine}
\label{originalgentle}
    $\T_1\oplus\T_2$ is decidable, when 
    $\T_1$ is gentle, and
    $\T_2$ is either: $(i)$~gentle, $(ii)$~finitely axiomatizable, or 
$(iii)$~there is an algorithm that, for a conjunction $\varphi$ of $\Sigma_2$-literals, outputs a finite $S \subset \N$ with $\spec(\T,\varphi) = S$.%
\footnote{Notice that in $(i)$, $S\subseteq\No$, and in $(ii)$ $S\subseteq\N$.}
\end{theorem}

\section{Limitations of theory combination methods}
\label{sec:infinitesigs}

In this section, we examine what is the outcome of dropping
each assumption on the theories from \Cref{originalNO,originalstrongpolite,originalshiny,originalgentle}.
We show that each of these theorems break if we drop any of the assumptions
it makes regarding the combined theories.

All of these theorems have the following form: if $\T_{1}$ and $\T_{2}$ are over disjoint signatures and are decidable, and in addition, $\T_{1}$ admits some properties, and $\T_{2}$ admits some properties, then $\T_{1}\oplus\T_{2}$ is also decidable.
Thus, our limitation proofs always consist of examples for theories $\T_{1}$ and $\T_{2}$
that admit all but one of the properties, such that $\T_{1}\oplus\T_{2}$ is undecidable.

In \Cref{sec:NOLimits} we show that the Nelson--Oppen combination method fails if we drop the requirement of stable infiniteness from one of the theories.
We show a similar result for gentle combination.
In \Cref{sec:POLimits}, dedicated to polite combination, we show that it 
fails if we drop any of the requirements for polite combination from $\T_2$.
This includes dropping strong finite witnessability in exchange for finite witnessability, namely, replacing strong politeness by politeness. 
We also show that dropping smoothness from the polite combination method results in 
failure.
In \Cref{sec:SHLimits}, a similar investigation is carried out for shiny theories.
For each of the three components of shininess, we show that it is critical for the possibility of combination.

But first, we introduce the theories that will be used to demonstrate the limits of the various combination theorems in~\Cref{sec:thetheories}. 

\subsection{The theories that we use}
\label{sec:thetheories}
Since the Nelson--Oppen and gentle methods require one property each (stable infiniteness and gentleness, respectively),
the polite method requires $2$ properties (smoothness and strong finite witnessability), 
and the shiny method requires $3$ properties (smoothness, the finite model property, and the computability of the minimal model function), we have
$1+1+2+3=7$ variants to consider, each removing exactly one property as an assumption from a combination theorem.
For each such variant, we need to provide $2$ theories, $\T_1$ and $\T_2$ for which the variant fails.
So, in total, we need to produce $7\cdot 2 = 14$ theories as examples.

Remarkably, we are able to cover all the aforementioned variants using only $4$ theories, that are defined over $3$ signatures.
Out of these $4$ theories, only $2$ are used for all but one of the variants.
From these $4$ theories, we create $3$ ordered pairs of theories ($\T_1$ and $\T_2$). Two pairs are used to show the limits of $3$ combination approaches each, and the third pair is used for one limit.
Clearly, even if we were only concerned with shininess, $3$ distinct ordered pairs would
have been necessary, as there are three properties to exclude.
Thus, the number of pairs of theories that we present is optimal.
%

The signatures for the theories are described in \Cref{tab:siginf}.
$\Sone$ is simply the empty signature. Atomic formulas
are therefore only equalities between variables.
$\Ss$ has a unary function symbol $s$.
And $\Spn$
has infinitely many $0$-ary predicate symbols $P_1,P_2,\ldots$.

\begin{table}[t]
\centering
\begin{tabular}{|c|c|c|}
\hline
Name & Function Symbols & Predicate Symbols \\\hline
$\Sone$ & $\emptyset$ & $\emptyset$ \\
$\Ss$ & $\{s\}$ & $\emptyset$\\
$\Spn$ & $\emptyset$ & $\{P_n\mid n\in\No\}$\\\hline
\end{tabular}
\vspace{2mm}
\caption{Signatures. Predicate symbols  are $0$-ary. The function symbol is unary.}
\label{tab:siginf}
\vspace{-8mm}
\end{table}

The $4$ theories are described in \Cref{table:theories}.
The first two are taken from \cite{CADE}, which introduced and studied a wide collection of theories.
$\Tinfty$ is the theory over the empty signature whose models have
infinitely many elements.

$\TM$ is more involved. 
Its axiomatization as a $\Ss$-theory assumes the existence of a non-computable function $f:\No\rightarrow\{0,1\}$,
such that $f(1)=1$, and for every $k\geq 1$, $f$ maps half of the numbers between $1$ and $2^k$ to $1$,
and the other half to $0$.
Such a function was proven to exist in \cite[Lemma~6]{CADE}.
The axiomatization utilizes two derived functions:
$f_0(k)$ returns the number of numbers between $1$ and $k$ that $f$ maps to $0$,
while
$f_1(k)$ returns the number of numbers between $1$ and $k$ that $f$ maps to $1$.
Obviously, when $k$ is a power of $2$, then $f_0(k)=f_1(k)$.
Now, $f$ itself is not a part of the signature $\Ss$ of $\TM$.
Instead, 
the axiomatization relies on the formulas from \Cref{fig-card-s},
that involve counting elements for which the function symbol $s$ acts as the identity.
Intuitively, a finite $\TM$-interpretation $\A$ with $n$ elements has $\of_{0}(n)$ of them satisfying  $\os^{\A}(e)\neq e$, and $\of_{1}(n)$ satisfying  $\os^{\A}(e)=e$;
an infinite such interpretation has infinitely many elements of each kind.

\begin{figure}[t]
\begin{mdframed}
\vspace{-1mm}
\small\[\psi^{\eqs}_{\geq n}=\Exists{\overarrow{x}}[\distinct{n}\wedge\bigwedge_{i=1}^{n}p(x_{i}) ],\quad\quad\psi^{\difs}_{\geq n}=\Exists{\overarrow{x}}[\distinct{n}\wedge\bigwedge_{i=1}^{n}\neg p(x_{i})],\]
\[\psi^{\eqs}_{=n}=\Exists{\overarrow{x}}[\distinct{n}\wedge\bigwedge_{i=1}^{n}p(x_{i})\wedge \Forall{x}[p(x)\rightarrow\bigvee_{i=1}^{n}x=x_{i}]],\]
\[\psi^{\difs}_{=n}=\Exists{\overarrow{x}}[\distinct{n}\wedge\bigwedge_{i=1}^{n}\neg p(x_{i})\wedge \Forall{x}[\neg p(x)\rightarrow\bigvee_{i=1}^{n}x=x_{i}]].\]
\vspace{-2mm}
\end{mdframed}
\caption{Formulas for the axiomatization of $\TM$. $\overarrow{x}$ stands for $x_{1},\ldots,x_{n}$, and $p(x)$ for $\os(x)=x$.}
\label{fig-card-s}
\vspace{0mm}
\end{figure}



The definition of $\Tgr$ assumes an arbitrary
non-computable function
$F:\No\rightarrow\No\cup\{\aleph_0\}$ such that
the set $\{(m,n)\in\No\times\No\mid F(m)\geq n\}$ is decidable.
Such a function $F$ exists: for example,
suppose $F$ maps every $n\in\No$ to the number of steps
the $n$th Turing machine (under some encoding) takes to halt, returning $\aleph_0$ if it does not halt. 
This function is clearly not computable. But, 
given $m$ and $n$, we can decide whether 
$F(m)\geq n$ by executing the $m$th Turing machine for $n$ steps.
If a $\Tgr$-interpretation $\A$ satisfies $P_{n}$, then it has at most $F(n)$ elements.\footnote{$\Tgr$ generalizes the theory $TM_{\infty}$ from \cite{Bonacina}.}

Finally, $\Teq$ consists of all $\Spn$-interpretations $\A$ in which for all $n\in\No$, 
either $P_n$ is interpreted as false, or $|\A|=n$.
It therefore allows quantifier-free formulas to enforce finite sizes of models, as $P_{n}$ being true implies the model has $n$ elements.

\begin{table}[t]
\centering
\begin{tabular}{|c|c|c|c|}
\hline
Name & Signature & Axiomatization & Source \\\hline
$\Tinfty$  & $\Sone$ & $\{\psi_{\geq n} : n\in\No\}$ & \cite{CADE} \\
$\TM$ & $\Ss$ & 
$\{[\psi^{=}_{\geq\of_{1}(k)}\wedge \psi^{\neq}_{\geq\of_{0}(k)}]$ 
 $\vee\bigvee_{i=1}^{k}[\psi^{=}_{=\of_{1}(i)}\wedge \psi^{\neq}_{=\of_{0}(i)}]: k\in\No\}$ &  \cite{CADE}\\
$\Tgr$ & $\Spn$ & $\{P_{n}\rightarrow \psi_{\leq F(n)} : \text{$n\in\No$, $F(n)\in\No$}\}$ & new \\ 
$\Teq$ & $\Spn$ & $\{P_{n}\rightarrow\psi_{=n}:n\in\No\}$ & new \\
\hline
\end{tabular}
\vspace{4mm}
\caption{Theories. 
$\of:\No\rightarrow\{0,1\}$ is assumed to be a non-computable function, such that $\of(1)=1$ and, for every $k\geq 0$,
$\of$ maps half of the numbers between $1$ and $2^k$ to $1$, and the other half to $0$. 
$f_i(k)$ is the number of numbers between $1$ and $k$ that are mapped by $f$ to $i$.
%
%
%
$F:\No\rightarrow\No\cup\{\Inf\}$ is non-computable,
but the set $\{(m,n)\mid F(m)\geq n\}$ is decidable. 
Formulas from
\Cref{fig-card-s} are used. 
}
\label{table:theories}
\vspace{-4mm}
\end{table}

\subsection{Nelson--Oppen and gentle combination}
\label{sec:NOLimits}
\label{sec:GELimits}

We begin by proving the sharpness of \Cref{originalNO} in the following sense:
    although two theories can be combined if both are stably-infinite, this is no longer the case if only one has that property.
    This result was previously proven in \cite[Theorem~4.1]{Bonacina}, 
    but with a different proof.

\begin{theorem}\label{Nelson Oppen sharpness}
There are decidable theories $\T_1$ and $\T_2$ 
over disjoint signatures such that $\T_1$ is stably infinite
but $\T_1\oplus\T_2$ is undecidable.
\end{theorem}

\begin{proof}[sketch]\footnote{Due to lack of space, some proofs are omitted, and can be found in the appendix.}
Take $\T_1$ and $\T_2$ to be $\TM$ and $\Teq$, respectively.
Clearly, their signatures (namely $\Ss$ and $\Spn$) are disjoint. 
Further, $\TM$ is shown in \cite[Lemma~54]{arxivCADE} to be stably infinite;
it was also proven to have the same
set of quantifier-free satisfiable formulas as the theory of an uninterpreted unary function, which makes it decidable.
Finally,
    although both $\TM$ and $\Teq$ are decidable, $\TM\oplus\Teq$ is not.
Indeed, the formulas $P_{n+1} \land \varphi_{\geq\of_1(n)+1}^{\eqs}$, where
\[
    \phieqs \coloneqq \bigwedge_{1\leq i<j\leq n}\neg(x_{i}=x_{j})\wedge\bigwedge_{i=1}^{n}\os(x_{i})=x_{i},
\]
are $\Teq \oplus \TM$-satisfiable if and only if $\of(n+1)=1$, whereas $f$ is a non-computable function.
\end{proof}

Clearly, \Cref{originalNO} and the proof of \Cref{Nelson Oppen sharpness}
imply that $\Teq$ is not stably infinite.
And indeed, for every $n$,
the formula $P_n$ is $\Teq$-satisfiable,
but only by a finite model.

As it turns out, the same theories can be used
to show a similar result for gentleness.

\begin{theorem}
\label{theo:gentsharp}
There are decidable theories $\T_1$ and $\T_2$ 
over disjoint signatures such that $\T_1$ 
is gentle,
but $\T_1\oplus\T_2$ is undecidable.
\end{theorem}

\begin{proof}[sketch]
We reuse the proof of \Cref{Nelson Oppen sharpness},
but flip the roles of the theories.
Now, we set $\T_1$ to be $\Teq$ and $\T_2$ to be $\TM$.
$\TM$ and $\Teq$ are both decidable, are over disjoint signatures,
but $\Teq\oplus\TM$ is undecidable.
The only thing left to show is that $\Teq$ is gentle,
which indeed can be shown.
\end{proof}

\Cref{originalgentle} and the proof of \Cref{theo:gentsharp} tell us
that $\TM$ is not gentle. And indeed,
 were $\TM$ gentle, one would be able to calculate $\of$.
 Similarly, $\TM$ does not satisfy any of the other two requirements from \Cref{originalgentle}.

\subsection{Polite combination}
\label{sec:POLimits}
\Cref{originalstrongpolite} demands two properties from $\T_2$ in order for it to be
combinable with any decidable theory $\T_1$ over a disjoint signature:
strong finite witnessability and smoothness.
We start by showing that if smoothness is removed from the requirements,
the theorem fails. 

\begin{theorem}
\label{strong politeness sharpness SFW}
There are decidable theories $\T_1$ and $\T_2$ 
over disjoint signatures such that $\T_2$ is strongly finitely witnessable
but $\T_1\oplus\T_2$ is undecidable.
\end{theorem}

\begin{proof}[sketch]
Take $\T_1$ to be $\TM$ and $\T_2$ to be $\Teq$,
as was done in the proof of \Cref{Nelson Oppen sharpness},
where both theories were shown to be decidable while their combination was shown to be undecidable.
The only thing that is left to be shown, and indeed can be shown by providing an appropriate strong witness, is that $\Teq$ is strongly finitely witnessable.
\end{proof}

As before, \Cref{originalstrongpolite} and the proof of \Cref{strong politeness sharpness SFW} imply
that $\Teq$ is not smooth. And indeed, it is not, as it is not even
stably infinite.

Next, we show that dropping the strong finite witnessability requirement also leads to a failure in the polite combination method.

\begin{theorem}
\label{strong politeness sharpness SM}
There are decidable theories $\T_1$ and $\T_2$ 
over disjoint signatures such that $\T_2$ is smooth
but $\T_1\oplus\T_2$ is undecidable.
\end{theorem}

\begin{proof}
%
Take $\T_1$ to be $\Teq$ and $\T_2$ to be $\TM$,
again as in the proof of \Cref{theo:gentsharp}, only now
we rely on the fact that, proven in
\cite[Lemma~54]{arxivCADE}, that $\TM$ is smooth, and as we already know, $\Teq\oplus \TM$ is not decidable. 
\end{proof}

Clearly, \Cref{originalstrongpolite} and the proof of \Cref{polite sharpness} imply that
$\TM$ is not strongly finitely witnessable. This was also proven in~\cite[Lemma~56]{arxivCADE}.

Now, $\TM$ was proven in \cite[Lemmas~55]{arxivCADE} to not be smooth, but it is also finitely witnessable (without being strongly finitely witnessable), which makes it polite.
%
%
Thus, the proof of \Cref{strong politeness sharpness SM} also gives us the following corollary, by again taking $\T_1$ to be $\Teq$ and $\T_2$ to be $\TM$.
\begin{corollary}
\label{polite sharpness}
There are decidable theories $\T_1$ and $\T_2$ 
over disjoint signatures such that $\T_2$ is polite
but $\T_1\oplus\T_2$ is undecidable.
\end{corollary}

 Recall that \cite{RanRinZar} claimed that politeness is enough
 for theory combination, but a problem in the proof was later discovered
 and corrected in \cite{JB10-LPAR} by strengthening the politeness assumption to
 strong politeness.
But was the problem of \cite{RanRinZar} in the proof or in the statement itself?
In other words: does \Cref{originalpolite} hold?
What we immediately get from \Cref{polite sharpness} is that it does not.

\begin{corollary}
\Cref{originalpolite} does not hold.
\end{corollary}

Hence, politeness is not enough for theory combination, which justifies the title of this paper.


\subsection{Shiny combination}
\label{sec:SHLimits}

In this section we consider the three requirements \Cref{originalshiny} makes
on one of the combined theories, namely: 
computability of the minimal model function, the finite model property, and smoothness.

We start with the computability of the minimal model function.

\begin{theorem}\label{shiny sharpness -CMMF}
There are decidable theories $\T_1$ and $\T_2$ 
over disjoint signatures such that $\T_2$ is 
smooth and has the finite model property,
but $\T_1\oplus\T_2$ is undecidable.
\end{theorem}

\begin{proof}
By taking $\T_1$ to be $\Teq$ and $\T_2$ to be $\TM$,
we can use proofs of previous theorems in order to show
most properties that are needed.
Further, it was proven in \cite[Theorem~2]{FroCoS} that
$\TM$ admits the finite model property.
\end{proof}

From \Cref{originalshiny} and the proof of \Cref{shiny sharpness -CMMF},
$\TM$ does not have a computable minimal model function, which was also proven in \cite[Lemma~126]{LPAR-arXiv}.

For the next sharpness theorem we need the following lemma,
according to which for decidable theories,
strong finite witnessability implies computability of the minimal model function.
This was essentially proven in \cite{CasalRasga2}, but was never explicitly stated there;
indeed, as they were focused on strong politeness and shininess, they have assumed smoothness, even if that assumption was never actually used in the part of the proof that concerned the computability of the minimal model function.

\begin{restatable}{lemma}{lemSFWandDECimpCMMF}
\label{lem:SFWandDECimpCMMF}
    If $\T$ is decidable and strongly finitely witnessable, then it has a computable minimal model function.
\end{restatable}


\begin{remark}
    Notice that the reciprocal of \Cref{lem:SFWandDECimpCMMF} is not true: decidability and computability of the minimal model function do not entail strong finite witnessability. 
    For example, $\Teven$, defined in \cite{SZRRBT-21} by the axiomatization $\{\neg\psi_{=2\cdot n+1}\mid n\in\mathbb{N}\}$, is proven in \cite[Lemma~126]{LPAR-arXiv} to have a computable minimal model function;
    furthermore it is decidable, as it satisfies all and only the quantifier-free formulas
    that are satisfiable in first-order logic, but it is not strongly finitely witnessable (as proven in \cite{SZRRBT-21}).
\end{remark}

Now, using \Cref{lem:SFWandDECimpCMMF}, 
we show that shiny combination (\Cref{originalshiny}) fails without the smoothness requirement.
We once again essentially reuse \Cref{Nelson Oppen sharpness} to obtain the following:

\begin{theorem}\label{ex:SHINYminusSM}
There are decidable theories $\T_1$ and $\T_2$ 
over disjoint signatures such that $\T_2$ 
has the finite model property and a computable minimal model function,
but $\T_1\oplus\T_2$ is undecidable.
\end{theorem}

\begin{proof}
Take $\T_1$ and $\T_2$ to be
$\TM$ and $\Teq$, respectively:
we have already shown that they are both decidable even though $\TM\oplus\Teq$ is not.
From \Cref{lem:SFWandDECimpCMMF} and the fact that $\Teq$ is strongly finitely witnessable (which was established in the proof of \Cref{strong politeness sharpness SFW}), 
we get $\Teq$ has a computable minimal model function.
Using then \cite[Theorem~2]{FroCoS}, 
according to which finite witnessability implies the finite model property, 
 $\Teq$ has the finite model property.
\end{proof}


Next, we show that the requirement of the finite model property cannot be removed.
Unlike the previous results, we are unable to reuse $\Teq$ and $\TM$.
Therefore, we use the theory $\Tinfty$ from \cite{CADE} and
the theory $\Tgr$.

\begin{theorem}\label{theo:SHINYminusSF}
There are decidable theories $\T_1$ and $\T_2$ 
over disjoint signatures such that $\T_2$ 
is smooth and has a computable minimal model function,
but $\T_1\oplus\T_2$ is undecidable.
\end{theorem}

\begin{proof}[sketch]
Take $\T_1$ to be $\Tgr$ and $\T_2$ to be $\Tinfty$,
axiomatized in~\Cref{table:theories}.
Clearly, they are defined over disjoint signatures.
$\Tinfty$ is smooth and has a computable minimal model function.
The proofs for these facts are simple, and are given in \cite[Lemma~22]{arxivCADE} and \cite[Lemma~130]{LPAR-arXiv}.
    It is also decidable, as it satisfies all quantifier-free formulas in its signature that are satisfiable in first-order logic (and only them).
Perhaps surprisingly, it is possible to show that $\Tgr$ is decidable.
However, it can also be shown that $\Tgr\oplus\Tinfty$ is not.
\end{proof}

\section{Finite signatures}\label{sec:finitesig}

Every proof in \Cref{sec:infinitesigs} uses a pair of theories,
one of them always over the infinite signature $\Spn$
(in all cases the used theory is $\Teq$, except for in \Cref{theo:SHINYminusSF}, 
where the $\Spn$-theory $\Tgr$ is used instead).
And indeed, both $\Teq$ and $\Tgr$ are theories that are relatively easy to understand.
This is, among other things, thanks to the availability of infinitely many predicates.

In this section, we aim to provide finitistic proofs of the limitation theorems from
\Cref{sec:infinitesigs}, in the sense that all theories that are used are
over finite signatures.
Doing so provides a more succinct set of examples, over more minimal signatures.
The cost, however, is that the theories that we use in this section are more complex.

\subsection{New theories over finite signatures}
The two theories over
the infinite signature $\Spn$ from \Cref{sec:infinitesigs}
are $\Teq$ and $\Tgr$.
They will be replaced by theories over the finite signature $\Sat$:
this signature has a unary function $\ns$, a constant $a$, and no predicates, as described in \Cref{tabfinitesig}.

\begin{table}[t]
\centering
\begin{tabular}{|c|c|c|c|}
\hline
Name & $0$-ary Functions & $1$-ary Functions & Predicates \\\hline
$\Sat$ & $\{a\}$ & $\{\ns\}$ & $\emptyset$\\\hline
\end{tabular}
\vspace{2mm}
\caption{A finite signature.}\label{tabfinitesig}
\end{table}

In order to introduce the new theories, we define the formulas $\dif{n}(x)$ and $\orb{n}(x)$ in \Cref{fig-orb-s-a}.
\begin{figure}[t]
\begin{mdframed}
\vspace{-1mm}
\[\dif{n}(x)=\bigwedge_{0\leq i<j\leq n-1}\neg(\ns^{i}(x)=\ns^{j}(x)) \text{ for } n\in\No\setminus\{1\}\]

\[\orb{1}(x)=\ns(x)=x\]

\[\orb{n}(x)=\dif{n}(x)\wedge\neg\dif{n+1}(x) \text{ for } n\geq 2\]
\normalsize
\end{mdframed}
\caption{Formulas in \Sat.}
\label{fig-orb-s-a}
\vspace{-2mm}
\end{figure}

The {\em orbit} (see e.g., \cite{Rotman,Kilp}) 
 of an element $e$ in a $\Sat$-interpretation $\A$ is the set 
$\{(\ns^{\A})^n(e)\mid n\in\mathbb{N}\}$.
Since $e$ itself is always an element of this set, the orbit is always
non-empty. 
We sometimes view this set as the following sequence indexed by $n$:
$e,\ns^{\A}(e),(\ns^{\A})^2(e),\ldots$
In this context, in an interpretation $\A$ that satisfies $\dif{n}(x)$, we have that there are at least $n$ elements that can be obtained by recursively applying $\ns^{\A}$ to $x^{\A}$, meaning its orbit has at least $n$ elements;
similar, if $\A$ satisfies $\orb{n}(x)$, the orbit of $x^{\A}$ has precisely $n$ elements.

With these formulas, we can now define the new theories over $\Sat$.
These are specified in \Cref{table:finitetheories}.
In the finite $\Tot$-interpretations, the orbit of the interpretation of the constant symbol $a$
consists of at least half of the elements of the interpretation;
meanwhile in the infinite $\Tot$-interpretations this orbit is infinite. 
$\Tgp$ is very similar to $\Tgr$, replacing $P_n$ by 
the assumption that the orbit has size $n$, and also
concluding that the number of elements in the domain is at most $F(n)+n$ (and not $F(n)$ as in $\Tgr$).

\begin{table}[t]
\centering
\begin{tabular}{|c|c|c|c|}
\hline
Name & Signature & Axiomatization & Source \\\hline
$\Tot$ & $\Sat$ & $\{\orb{n}(a)\rightarrow\psi_{\leq2n}:n\in\No\}$ & new \\
$\Tgp$ & $\Sat$ & 
$\{\orb{n}(a)\rightarrow \psi_{\leq F(n)+n} : \text{$n\in\No$, $F(n)\in\No$}\}$ & new \\[.2em]
\hline
\end{tabular}
\vspace{2mm}
\caption{Theories over the finite signature $\Sat$.
In the definition of $\Tgp$, the function $F$ admits the same assumptions as in  \Cref{table:theories}.
The axiomatizations utilize formulas that are defined in
\Cref{fig-orb-s-a}.
}
\label{table:finitetheories}
\end{table}

With these new theories, we can now turn to making the proofs of the theorems from \Cref{sec:infinitesigs}
rely solely on 
finite signatures.

\subsection{Finitizing the proofs of \Cref{Nelson Oppen sharpness,theo:gentsharp,strong politeness sharpness SFW,ex:SHINYminusSM,strong politeness sharpness SM,shiny sharpness -CMMF,polite sharpness}}

The proof of \Cref{Nelson Oppen sharpness} 
sets $\T_1$ to be $\TM$ and $\T_2$ to be $\Teq$.
In order to only use finite signatures,
we set $\T_2$ to be $\Tot$ instead.\footnote{Notice that \cite{Bonacina} has also produced a finitary proof of \Cref{Nelson Oppen sharpness}, using a theory named $TM_{\forall\omega}$ over a finite signature.
}


As for $\T_1$, we can still use $\TM$, but
we need to restrict the possible functions $\of$ it relies on.
To make it clear that the functions $\of$ are now required to satisfy some extra properties we denote them by $\nft$, so that $\TM$ becomes $\TMn$. 
We then require 
$\nft:\No\rightarrow\{0,1\}$ to be any non-computable function such that:
$\nft(1)=1$ and $\nft$ is zero as often as it is $1$ in each interval from $1$ to $2^{k}$ (as required for $\of$ in the definition of $\TM$);
and in addition to the requirements in $\TM$, we now also require that $\nft(2n+1)=\nft(2n+2)$ for all $n\geq 2$.
Such functions exist:
for an example, take the function $\of:\No\rightarrow\{0,1\}$ defined in \cite{CADE}, make $\nft(1)=\nft(3)=1$, $\nft(2)=\nft(4)=0$, and $\nft(2n+1)=\nft(2n+2)=\of(n+1)$ for $n\geq 2$.
Since $\TM$ is decidable and stably infinite regardless of the specific $\of$, we have $\TMn$ is decidable and stably infinite.

Although it can be shown that $\Tot$ is decidable, $\TMn\oplus\Tot$ is undecidable.
Indeed, were it decidable, one would be able to calculate the function $\nft$ by using the fact that $\orb{n+1}(a)\wedge\varphi^{\eqs}_{\geq\nft_{1}(2n)+2}$ is $\TMn\oplus\Tot$-satisfiable if and only if $\nft(2n+1)=\nft(2n+2)=1$ for $n\geq 2$:
if we know $\nft$ up to $2n$ we can calculate $\nft_{1}(2n)$, obtain the formula $\orb{n+1}(a)\wedge\varphi^{\eqs}_{\geq\nft_{1}(2n)+2}$, and by testing whether it is $\TMn\oplus\Tot$-satisfiable we find the value for $\nft(2n+1)=\nft(2n+2)$. 
We know $\nft(1)=\nft(3)=1$ and $\nft(2)=\nft(4)=0$, and then we proceed from there on forward inductively.

We can mimic the same process for other results from \Cref{sec:infinitesigs},
by replacing $\Teq$ by $\Tot$, and instantiating $\TM$ by $\TMn$ with $\of$ satisfying the aforementioned condition.
In particular, we can do so in:
\Cref{theo:gentsharp},
by proving that $\Tot$ is also gentle;
\Cref{strong politeness sharpness SFW}, by proving that $\Tot$ is strongly finitely witnessable;
\Cref{strong politeness sharpness SM}, and \Cref{polite sharpness}, by remembering $\TMn$ is both smooth and polite;
\Cref{shiny sharpness -CMMF}, as $\TMn$ has the finite model property;
and \Cref{ex:SHINYminusSM}, by proving that $\Tot$ has a computable minimal model function, and the finite model property.

\subsection{Finitizing the proof of \Cref{theo:SHINYminusSF}}
Finally, notice that the proof of 
\Cref{theo:SHINYminusSF} sets $\T_1$ to be 
$\Tgr$, which is defined over an infinite signature,
and $\T_2$ to be $\Tinfty$.
While we can leave $\T_2$ as $\Tinfty$,
    we replace $\Tgr$  by its $\Sat$-variant $\Tgp$,
in order to get two theories over a finite signature.

    The proof that $\Tgp$ is decidable follows the proof that $\Tot$ is decidable.
    Yet the combination $\Tgp\oplus\Tinfty$ is not decidable, 
    $\orb{n}(a)$ being satisfiable in it if, and only if,  $F(n)=\Inf$.

\section{New combination theorems}
\label{sec:newtheorems}
In this section we prove new combination theorems, that strengthen
\Cref{originalgentle,originalshiny}.
In \Cref{sec:recovering-gentle}, we show that the conditions from \Cref{gentle-recovery-fontaine} can be weakened.
In \Cref{sec:recovering-shiny}, we show that the finite model property {\em can} be dismissed
from shiny combination, as long as we compensate it by requiring another property
from the second theory being combined.
This does not contradict \Cref{theo:SHINYminusSF}, as the example there does not meet the additional criterion.

In what follows, we assume that $\Sigma_1$ and $\Sigma_2$ are disjoint
signatures, and that $\T_1$ is a $\Sigma_1$-theory, and
$\T_2$ a $\Sigma_2$-theory.

\subsection{Recovering gentle combination}
\label{sec:recovering-gentle}

Assuming $\T_1$ is gentle, \Cref{gentle-recovery-fontaine} provided three conditions on $\T_2$, any one of which suffices for theory combination.
We prove a strengthening of \Cref{gentle-recovery-fontaine}.

\begin{definition}
    We say that a theory $\T$ has \emph{computable finite spectra} if there is an algorithm that, given a quantifier-free formula $\varphi$ and $k \in \No$, decides whether $k \in \spec(\T,\varphi)$.
\end{definition}

Intuitively, having computable finite spectra means that we can query the set $\spec(\T,\varphi)$ to check whether it contains a given {\em finite} cardinality. In contrast to gentleness, it does not imply that we can compute any concrete set $S$,
nor does it require 
the ability to check whether $\aleph_0$ is in the spectra.

\begin{theorem} \label{gentle-recovery}
    Suppose that $\T_1$ is gentle and $\T_2$ has computable finite spectra. Then, $\T_1 \oplus \T_2$ is decidable.
\end{theorem}

Each of the three properties in \Cref{gentle-recovery-fontaine} imply that $\T_2$ has computable finite spectra, so \Cref{gentle-recovery} is indeed a strengthening. 
%
We now present two theories that can be combined by \Cref{gentle-recovery} but not by 
any other combination method discussed in this paper.

\begin{example} \label{gentle-recovery-ex}
Fix any $n\in\No$ and
    let $\T_1 = \T_{\leq n}$ be the $\Sone$-theory axiomatized by $\{\psi_{\leq n}\}$. 
    Now,  let $\T_2$ be $\Tgr$ from \Cref{table:theories}. Then, $\T_{\leq n}$ and $\Tgr$ are decidable, $\T_{\leq n}$ is gentle, and $\Tgr$ has computable finite spectra. By \Cref{gentle-recovery}, $\T_{\leq n} \oplus \Tgr$ is decidable. On the other hand, $\Tgr$ does not satisfy any of the three properties in \Cref{gentle-recovery-fontaine}. Furthermore, neither theory is strongly polite, shiny, or stably infinite, and so none of the other combination theorems can be used to decide this combination of theories.
\end{example}

\subsection{Recovering shiny combination without finite models}
\label{sec:recovering-shiny}
In \Cref{theo:SHINYminusSF} we have seen that the shiny combination theorem fails if the finite model property is dropped from the definition of shininess. However, we now show that we can do without the finite model property if we impose another condition on the other theory being combined.

\begin{definition}
    We say that a theory $\T$ is \emph{infinitely decidable} if it is decidable whether a quantifier-free formula is satisfied by an infinite $\T$-interpretation. 
\end{definition}

A very similar notion to infinite decidability, that also requires the theory to be decidable, was defined in \cite{Bonacina} and called $\exists_\infty$-decidability, but not considered along gentleness.


\begin{theorem} \label{new-method}
    Let $\T_1$ and $\T_2$ be decidable theories over disjoint signatures. Suppose that $\T_1$ is smooth and has a computable minimal model function and that $\T_2$ is infinitely decidable. Then, $\T_1 \oplus \T_2$ is decidable.
\end{theorem}

In the next example, we present 2 theories that 
can be combined using \Cref{new-method},
but not with any other combination theorem studied in this paper.

\begin{example}
\label{theo:onlybencomb}
Let $\Tb$ be the $\Spn$-theory axiomatized by
\[\{P_{1}\rightarrow\psi_{=1}\}\cup\{P_{1}\rightarrow\neg P_{n} : n\geq 2\}\cup\{P_{n}\rightarrow\psi_{\geq m} : m,n\geq 2, \nf(n)=1\},\]
for $\nf:\No\rightarrow\{0,1\}$ a non-computable function.
Also, consider the theory $\Tinfty$ from \Cref{table:theories}.
Both theories are decidable.
It can be shown that 
$\Tb$ is neither stably infinite nor has computable finite spectra, and so it cannot be combined with $\Tinfty$ using the Nelson--Oppen method, the gentle method,
or the new method we propose in \Cref{gentle-recovery}.
It can also be shown that neither
 theory is strongly polite or shiny, and so they cannot be combined 
using the polite or shiny methods.
But, $\Tinfty$ has a computable minimal model function and is smooth.
Further, $\Tb$ is infinitely decidable.
By \Cref{new-method}, $\Tinfty\oplus\Tb$ is decidable.
Thus, \Cref{new-method}
is able to combine two theories that none of the other methods can.\footnote{
We use $\Tinfty$ in the example to keep things simple, reusing the theories that are already defined in the paper.
However, any decidable theory $\T$ on a countable signature (disjoint from $\Spn$) with only infinite models could replace $\Tinfty$, such as the theory of dense linear orders without endpoints \cite{Kozen}.
}
\end{example}

\section{Conclusion}
\label{sec:conclusion}
For each combination method and each of its associated properties, we have proven in \Cref{sec:infinitesigs} that the corresponding combination theorem fails if the property is not assumed. 
The proofs always involve producing two theories that 
are decidable while their combination is not.
The proofs of these results were improved in \Cref{sec:finitesig},
where only finite signatures were used.

\Cref{tab:summary} 
 lists the theories used in \Cref{sec:infinitesigs,sec:finitesig}.
It also lists the original combination theorem whose limits are identified.
Notice that for each theorem we produced 2 pairs of theories: one pair for its original proof, and another pair for its improved proof.
In total, we were able to prove all theorems, with finite and infinite signatures, using only three quadruples of theories, built from only
six theories, by reusing the introduced theories as much as possible.

We have also proven that 
politeness is not enough for theory combination.
Further, we have introduced two new combination theorems, based on shiny and gentle combinations (\Cref{new-method,gentle-recovery}).

The main direction for further work is to find more theorems like \Cref{new-method,gentle-recovery},
with the purpose of varying the set of requirements for
theory combination.
We hope that such theorems will make it 
to introduce algorithms for new combination of theories.
In addition, we are working on stronger limitation theorems: while
the classical combination methods provide {\em sufficient} conditions for combinability, we plan to study {\em necessary} conditions.

\paragraph{Acknowledgments}
We thank Christophe Ringeissen, Pascal Fontaine, and Cesare Tinelli for fruitful discussions that led to and helped  writing this paper.

\begin{table}[t]
\centering
\renewcommand{\arraystretch}{1.4}
\begin{tabular}{|c|c|c|c|c||c|c|}
\hline
\centering \multirow{2}{*}{Approach} & \multirow{2}{*}{Property} & \multirow{2}{*}{Theorem} &  \multicolumn{2}{c||}{Infinite} & \multicolumn{2}{c|}{Finite} \\
&&& $\T_1$ & $\T_2$ & $\T_1$ & $\T_2$ \\\hline\hline
\centering
Nelson--Oppen  (Thm. \ref{originalNO}) & Stable Infiniteness &  Thm. \ref{Nelson Oppen sharpness} & $\cellcolor{yellow!15}\TM$ & $\cellcolor{yellow!15}\Teq$ & $\cellcolor{yellow!15}\TMn$ & $\cellcolor{yellow!15}\Tot$\\\hline
\centering
Gentle  (Thm. \ref{originalgentle}) & Gentleness$^{\ast}$  & Thm. \ref{theo:gentsharp} &  $\cellcolor{blue!15}\Teq$ &
$\cellcolor{blue!15}\TM$
& $\cellcolor{blue!15}\Tot$
& $\cellcolor{blue!15}\TMn$ \\\hline
\multirow{2}{*}{\centering {\centering Polite  (Thm. \ref{originalstrongpolite})}} & Smoothness & Thm. \ref{strong politeness sharpness SFW} & $\cellcolor{yellow!15}\TM$ & $\cellcolor{yellow!15}\Teq$ & $\cellcolor{yellow!15}\TMn$ & $\cellcolor{yellow!15}\Tot$\\
& Strong Finite Witnessability$^{\ast\ast}$ & Thm. \ref{strong politeness sharpness SM} & 
$\cellcolor{blue!15}\Teq$ &
$\cellcolor{blue!15}\TM$ &
$\cellcolor{blue!15}\Tot$ & $\cellcolor{blue!15}\TMn$   \\
\hline
\multirow{3}{*}{\centering{\centering Shiny (Thm. \ref{originalshiny})}}  & Comp. Min. Mod. & Thm. 
\ref{shiny sharpness -CMMF} &  
$\cellcolor{blue!15}\Teq$ &
$\cellcolor{blue!15}\TM$ &
$\cellcolor{blue!15}\Tot$ & $\cellcolor{blue!15}\TMn$  \\
 & Smoothness & Thm. \ref{ex:SHINYminusSM} & 
 $\cellcolor{yellow!15}\TM$ & $\cellcolor{yellow!15}\Teq$ & 
 $\cellcolor{yellow!15}\TMn$ & $\cellcolor{yellow!15}\Tot$ \\
 & Finite Model Property$^{\ast}$ & Thm. \ref{theo:SHINYminusSF} & $\cellcolor{green!15}\Tgr$ & $\cellcolor{green!15}\Tinfty$ & $\cellcolor{green!15}\Tgp$ & $\cellcolor{green!15}\Tinfty$\\\hline
\end{tabular}
\renewcommand{\arraystretch}{1}
\vspace{2mm}
\caption{Summary of the main results. Only 3 tuples of theories were used, and these  are assigned different colors in the table. \\\vspace{-.5em}\\
$^{\ast}$ \Cref{new-method,gentle-recovery}, in a sense, remedy \Cref{theo:SHINYminusSF,theo:gentsharp}.\\
$^{\ast\ast}$ In particular, \Cref{polite sharpness} is a consequence of the proof of Thm. \ref{strong politeness sharpness SM}.%
}
\label{tab:summary}
\end{table}

\newpage

\bibliography{bib}{}
\bibliographystyle{plain}
\appendix


\pagebreak
\section{\tp{Proof of \Cref{Nelson Oppen sharpness}}{Proof of Result \ref{Nelson Oppen sharpness}}}


\begin{proposition}\label{Teq is gentle}
    The theory $\Teq$ is gentle.
\end{proposition}
\begin{proof}
    Let $\varphi$ be a conjunction of literals. Write $\varphi = \varphi_1 \land \varphi_2$, where $\varphi_1$ contains the equalities and disequalities in $\varphi$ and $\varphi_2$ contains the literals of the form $P_n$ and $\lnot P_n$ in $\varphi$.

    If $\varphi_1$ is unsatisfiable in equational logic, then $\spec(\Teq,\varphi) = \emptyset$; 
    otherwise, let $m$ be the size of the smallest interpretation that satisfies $\varphi_1$,
    which is possible
    since the theory of equality is shiny (see~\cite{DBLP:journals/jar/TinelliZ05}).

    If $\varphi_2$ contains no positive literals, then $\spec(\Teq,\varphi) = \{n \in \No \mid n \ge m\} \cup \{\aleph_0\}$. If $\varphi_2$ contains exactly one positive literal $P_n$, then $\spec(\Teq,\varphi) = \{n\}$. If $\varphi_2$ contains two positive literals $P_n$ and $P_{n'}$ where $n \neq n'$, then $\spec(\Teq,\varphi) = \emptyset$.

    Thus, we have shown that $\spec(\Teq,\varphi)$ is either a finite set of finite cardinalities that can be computed or a cofinite set whose complement can be computed; that is, $\Teq$ is gentle.
\end{proof}

\begin{proposition}
    The theory $\Teq \oplus \TM$ is undecidable.
\end{proposition}
\begin{proof}
    Let
    \[
        \phieqs \coloneqq \bigwedge_{1\leq i<j\leq n}\neg(x_{i}=x_{j})\wedge\bigwedge_{i=1}^{n}\os(x_{i})=x_{i}.
    \]
    Given that $\of$ is non-computable, it suffices to show that for each $n \in \No$, the sentence $P_{n+1} \land \varphi_{\geq\of_1(n)+1}^{\eqs}$ is $\Teq \oplus \TM$-satisfiable if and only if $\of(n+1)=1$; if $\Tp \oplus \TM$ were decidable, this would allow us to compute $\of(n+1)$ recursively in terms of $\of(1), \dots, \of(n)$.

    First, suppose $\of(n+1)=1$. Then, $\of_1(n+1) =\of_1(n)+1$. Since $\TM$ is smooth \cite[Lemma~54]{arxivCADE} and has an interpretation of size 1, there is a $\TM$-interpretation of every size in $\No$. Further, any $\TM$-interpretation of size $m \in \No$ satisfies $\varphi_{\geq\of_1(m)}^{\eqs}$. Thus, there is a $\TM$-interpretation $\A$ of size $n+1$ satisfying $\varphi_{\geq\of_1(n)+1}^{\eqs}$. We can extend $\A$ to a $\Teq \oplus \TM$-interpretation $\B$ satisfying $P_{n+1} \land \varphi_{\geq\of_1(n)+1}^{\eqs}$ by letting $P_{n+1}^\B$ be true and $P_{n'}^\B$ be false for all $n' \neq n+1$.

    Second, suppose $\of(n+1)=0$. Then, $\of_1(n+1) =\of_1(n)$, so any $\TM$-interpretation $\A$ of size $n+1$ satisfies $\psi^{\eqs}_{=\of_1(n)}$ and therefore does not satisfy $\varphi_{\geq\of_1(n)+1}^{\eqs}$. Hence, $P_{n+1} \land \varphi_{\geq\of_1(n)+1}^{\eqs}$ is $\Teq \oplus \TM$-unsatisfiable.
\end{proof}

\section{\tp{Proof of \Cref{theo:gentsharp}}{Proof of Result \ref{theo:gentsharp}}}

It suffices to show that $\Teq$ is gentle, and this was done in \Cref{Teq is gentle}.

\section{\tp{Proof of \Cref{strong politeness sharpness SFW}}{Proof of Result \ref{strong politeness sharpness SFW}}}

\begin{proposition}
    The theory $\Teq$ is strongly finitely witnessable.
\end{proposition}
\begin{proof}
    It suffices to define a strong witness for $\Spn$-formulas that are conjunctions of literals, so let $\varphi$ be a conjunction of literals. If $\varphi$ does not contain any literals of the form $P_n$, then let $\wit(\varphi) \coloneq \varphi \land w = w$ (where $w$ is fresh). Otherwise, let $n$ be the largest natural number such that the literal $P_n$ is in $\varphi$. Then, let
    \[
        \wit(\varphi) \coloneqq \varphi \land \bigwedge_{1 \le i < j \le n} w_i \neq w_j,
    \]
    where each $w_i$ is a fresh variable.

    First, we show that $\varphi$ and $\Exists{\overarrow{w}} \wit(\varphi)$ are $\Teq$-equivalent, where $\overarrow{w} = \vars(\wit(\varphi)) \setminus \vars(\varphi)$. This is clear if $\varphi$ does not contain any literals of the form $P_n$. Otherwise, let $n$ be the largest natural number such that the literal $P_n$ is in $\varphi$. Since $\Teq$ has the axiom $P_n \rightarrow \psi_{=n}$, any $\Teq$-interpretation that satisfies $\varphi$ has $n$ elements. In particular, any $\Teq$-interpretation that satisfies $\varphi$ satisfies
    \[
        \Exists{\overarrow{w}} \bigwedge_{1 \le i < j \le n} w_i \neq w_j.
    \]
    It follows that $\varphi$ and $\Exists{\overarrow{w}} \wit(\varphi)$ are $\Teq$-equivalent.

    Now, let $\delta$ be an arrangement on a finite set of variables $V$ such that $\wit(\varphi) \land \delta$ has a $\Teq$-interpretation $\A'$ satisfying it. 
    We need to show that there is a $\Teq$-interpretation $\A$ satisfying $\wit(\varphi) \land \delta$ such that $\dom{\A} = \vars(\wit(\varphi) \land \delta)^\A$. If $\varphi$ does not contain any literals of the form $P_n$, then we get our desired interpretation $\A$ by letting $\dom{\A} = \vars(\wit(\varphi) \land \delta)^{\A'}$, letting $x^\A = x^{\A'}$ for each variable $x \in \vars(\wit(\varphi) \land \delta)$ (and letting $x^\A$ be arbitrary for $x \notin \vars(\wit(\varphi) \land \delta)$), and letting $P_n^\A$ be false for all $n \in \No$. Otherwise, let $n$ be the (necessarily unique) natural number such that the literal $P_n$ is in $\varphi$. Then, $\A'$ has exactly $n$ elements, so $\dom{\A'} = \{w_1^{\A'}, \dots, w_n^{\A'}\}$. Thus, we can simply take $\A = \A'$ in this case.
\end{proof}


\section{\tp{Proof of \Cref{lem:SFWandDECimpCMMF}}{Proof of Result \ref{lem:SFWandDECimpCMMF}}}
We actually prove this in the many-sorted setting,\footnote{See \cite{CADE} for the definitions of a many-sorted signature, and what an interpretation is in that case.}
and then \Cref{lem:SFWandDECimpCMMF} follows as a particular instance,
where the number of sorts is 1.
As mentioned, this was already proven in \cite{CasalRasga2} (also for the many-sorted case),
but not explicitly stated.
More precisely, there, they proved shininess from strong politeness, and so relied on smoothness.
However, a careful look at the proof reveals the fact that smoothness was not relied on in the specific part of the proof that showed the computability of the minimal model function.
To be safe, we provide a full proof here.

Let $\Sigma$ be a many-sorted, first-order signature, $S$ a finite set of its sorts, and $\vars_{\s}(\varphi)$ be the set of variables of sort $\s$ in $\varphi$.

\paragraph{Finite witnessability:}
A $\Sigma$-theory $\T$ is finitely witnessable \wrt $S$ if there exists a function $\wit$ (called a witness) from the quantifier-free formulas of $\Sigma$ into themselves such that, for every quantifier-free formula $\varphi$, one has that: 
$(I)$ $\varphi$ and $\Exists{\overarrow{x}}\wit(\varphi)$ are $\T$-equivalent, where $\overarrow{x}=\vars(\wit(\varphi))\setminus\vars(\varphi)$;
$(II)$ if $\wit(\varphi)$ is $\T$-satisfiable there exists a $\T$-interpretation $\A$ that satisfies $\wit(\varphi)$ with $\s^{\A}=\vars_{\s}(\wit(\varphi))^{\A}$.

\paragraph{Strong finite witnessability:}
A $\Sigma$-theory $\T$ is strongly finitely witnessable \wrt $S$ if it is finitely witnessable \wrt $S$, with witness $\wit$, which in addition satisfies:
$(II^{\prime})$ for every quantifier-free formula $\varphi$, finite set of variables $V$, and arrangement $\delta_{V}$ on $V$, if $\wit(\varphi)\wedge\delta_{V}$ is $\T$-satisfiable then there exists a $\T$-interpretation $\A$ that satisfies $\wit(\varphi)\wedge\delta_{V}$ with $\s^{\A}=\vars_{\s}(\wit(\varphi)\wedge\delta_{V})^{\A}$.
In that case $\wit$ is called a strong witness.

\paragraph{Minimal model function}
The minimal model function $\minmod$ \wrt $S$ of a $\Sigma$-theory $\T$ is a function from the quantifier-free formulas of $\Sigma$ to the power set of $\N^{S}$ (that is, the set of functions from $S$ to $\N$) such that:
$(I)$ if $\varphi$ is a quantifier-free $\T$-satisfiable formula and $\textbf{n}\in\minmod(\varphi)$, there exists a $\T$-interpretation $\A$ that satisfies $\varphi$ with $|\s^{\A}|=\textbf{n}(\s)$ for every $\s\in S$;
$(II)$ if $\varphi$ is a quantifier-free $\T$-satisfiable formula, $\textbf{n}\in\minmod(\varphi)$, and $\B$ is a $\T$-interpretation that satisfies $\varphi$ with $|\s^{\B}|\neq \textbf{n}(\s)$ for some $\s\in S$, there exists $\s_{*}\in S$ such that $\textbf{n}(\s_{*})<|\s_{*}^{\B}|$.

\begin{theorem}
    If $\T$ is decidable and strongly finitely witnessable with respect to a finite set of sorts $S$, then it has a computable minimal model function with respect to $S$.
\end{theorem}

\begin{proof}

    Assume, without loss of generality, that $S=\{\s_{1},\ldots,\s_{n}\}$ so that we may write an element of the minimal model function as $(n_{1},\ldots,n_{n})$.
    Let $\wit$ be the strong witness for $\T$, $V$ be the set of variables in $\wit(\varphi)$, $V_{i}$ the set of variables in $\wit(\varphi)$ of sort $\s_{i}$, $Eq_{i}(V)$ the set of equivalence relations on $V_{i}$ (finite and easily algorithmically found), and $\eq{V}$ the product of $Eq_{i}(V)$.
    We then define $\minmod(\varphi)$ as the set of minimal elements of the set
    \[T(\varphi)=\{(|V_{1}/E_{1}|,\ldots,|V_{n}/E_{n}|) : \text{$E\in\eq{V}$ and $\wit(\varphi)\wedge\delta_{V}^{E}$ is $\T$-satisfiable}\},\]
    under the order such that $(n_{1},\ldots,n_{n})\leq (m_{1},\ldots,m_{n})$ if and only if $n_{i}\leq m_{i}$ for each $1\leq i\leq n$, and where:
    $E=(E_{1},\ldots,E_{n})$, and $\delta_{V}^{E}$ is the arrangement on $V$ inducing the equivalence $E_{i}$ on each $V_{i}$.
    This is computable as $\wit$ is computable, $\T$ is decidable, and the set whose minimal elements we must find is finite.

    Take an element $(|V_{1}/E_{1}|,\ldots,|V_{n}/E_{n}|)$ of $\minmod(\varphi)$, meaning $\wit(\varphi)\wedge\delta_{V}^{E}$ is $\T$-satisfiable;
    as $\wit$ is a strong witness, there is a $\T$-interpretation $\A$ that satisfies $\wit(\varphi)\wedge\delta_{V}^{E}$ with $\s_{i}^{\A}=\vars_{\s_{i}}(\wit(\varphi)\wedge\delta_{V}^{E})^{\A}=V_{i}^{\A}$ for every $\s_{i}\in S$. 
    Since $\A$ satisfies $\delta_{V}^{E}$, $V_{i}^{\A}$ has as many elements as $V_{i}/E_{i}$, and so $(|\s_{1}^{\A}|,\ldots,|\s_{n}^{\A}|)=(|V_{1}/E_{1}|,\ldots,|V_{n}/E_{n}|)$, meaning the first property of a minimal model function is satisfied.
    
    Now, suppose for the sake of contradiction that there is a tuple $(|V_{1}/E_{1}|,\ldots,\break|V_{n}/E_{n}|)$ in $\minmod(\varphi)$ and a $\T$-interpretation $\A$ that satisfies $\varphi$ such that $|\s_{i}^{\A}|\leq |V_{i}/E_{i}|$, for all $\s_{i}\in S$, and for at least one of them $|\s_{i}^{\A}|<|V_{i}/E_{i}|$.
    As $\wit$ is a strong witness we have that $\A$ satisfies $\Exists{\overarrow{x}}\wit(\varphi)$, for $\overarrow{x}=\vars(\wit(\varphi))\setminus\vars(\varphi)$, and therefore there is a $\T$-interpretation $\B$, differing from $\A$ at most on the values assigned to $\overarrow{x}$, that satisfies $\wit(\varphi)$.
    Let $F_{i}$ be the equivalence relation induced by $\B$ over $V_{i}$, and $F=(F_{1},\ldots,F_{n})$, so that $\B$ satisfies $\wit(\varphi)\wedge\delta_{V}^{F}$.
    Again, by using the fact that $\wit$ is a strong witness there must exist a third $\T$-interpretation $\C$ that satisfies $\wit(\varphi)\wedge\delta_{V}^{F}$ with $\s_{i}^{\C}=\vars_{\s_{i}}(\wit(\varphi)\wedge\delta_{V}^{F})^{\C}=V_{i}^{\C}$ for each $1\leq i\leq n$.
    Using $\C$ satisfies $\delta_{V}^{F}$, $|\s_{i}^{\C}|=|V_{i}/F_{i}|$, and since $\B$ also satisfies $\delta_{V}^{F}$, $|V_{i}/F_{i}|\leq |\s_{i}^{\B}|=|\s_{i}^{\A}|$.
    This means that $(|\s_{1}^{\C}|,\ldots,|\s_{n}^{\C}|)$, although being in $T(\varphi)$, is strictly less than $(|V_{1}/E_{1}|,\ldots,|V_{n}/E_{n}|)$, a minimal element of $T(\varphi)$, leading to a contradiction and finishing the proof.
\end{proof}

\section{\tp{Proof of \Cref{theo:SHINYminusSF}}{Proof of Result \ref{theo:SHINYminusSF}}}

\begin{proposition} \label{Tgr-dec}
    The theory $\Tgr$ is decidable.
\end{proposition}
\begin{proof}
    It suffices to show that it is decidable whether a conjunction of literals is $\Tgr$-satisfiable, so let $\varphi$ be a conjunction of literals. Write $\varphi = \varphi_1 \land \varphi_2$, where $\varphi_1$ contains the equalities and disequalities in $\varphi$ and $\varphi_2$ contains the literals of the form $P_n$ and $\lnot P_n$ in $\varphi$.

    We describe our decision procedure as follows. If $\varphi_1$ is unsatisfiable in equational logic, then $\varphi$ is $\Tgr$-unsatisfiable. Otherwise, let $m$ be the size of the smallest interpretation that satisfies $\varphi_1$.

    We claim that, in this case, $\varphi$ is $\Tgr$-satisfiable if and only if for every $n$ such that the literal $P_n$ is in $\varphi_2$, we have $F(n) \ge m$. This is because if the latter condition holds, we can extend an interpretation that satisfies $\varphi_1$ to a $\Tgr$-interpretation $\A$ satisfying $\varphi$ by setting $P_n^\A$ to true for every $n$ such that the literal $P_n$ is in $\varphi_2$ and setting $P_n^\A$ to false otherwise. Otherwise, there is some $n$ such that the literal $P_n$ is in $\varphi_2$ and $F(n) < m$. In this case, there is no interpretation that satisfies $\varphi_1$ of size at most $F(n)$, so $\varphi$ is $\Tgr$-unsatisfiable (since $P_n \rightarrow \psi_{\leq F(n)}$ is an axiom of $\Tgr$).
\end{proof}

\begin{proposition}
        The theory $\Tinfty$ is decidable.
    \end{proposition}

    \begin{proof}
    We prove $\Tinfty$ and equational logic satisfy the same quantifier-free formulas, and since the latter is decidable so will be $\Tinfty$.
    Of course equational logic satisfies all quantifier-free formulas that $\Tinfty$ satisfies, given that it has more models than $\Tinfty$.
    Reciprocally, suppose the quantifier-free formula $\varphi$ is satisfied by equational logic, and let $\A$ be an interpretation that satisfies $\varphi$.
    We consider the interpretation $\B$ with $\dom{\B}=\dom{\A}\cup B$, for a set $B=\{b_{n}:n\in\mathbb{N}\}$ disjoint from $\dom{\A}$, and $x^{\B}=x^{\A}$ for all variables $x$.
    $\B$ is a $\Tinfty$-interpretation, of course, but it also satisfies $\varphi$, which can be proven by a simple induction on the subformulas of $\varphi$.
    \end{proof}

\begin{proposition}
    The theory $\Tinfty \oplus \Tgr$ is undecidable.
\end{proposition}
\begin{proof}
    It suffices to show that $P_n$ is $\Tinfty \oplus \Tgr$-satisfiable if and only if $F(n) = \Inf$. If $F(n) = \Inf$, then $P_n$ is satisfied by the $\Tinfty \oplus \Tgr$-interpretation $\A$ of size $\aleph_0$ where $P_n^\A$ is true and $P_{n'}^\A$ is false for all $n' \neq n$. If $F(n) < \Inf$, then any $\Tgr$-interpretation satisfying $P_n$ must be finite (since $\Tgr$ has the axiom $P_n \rightarrow \psi_{\leq F(n)}$). Hence, $P_n$ is $\Tinfty \oplus \Tgr$-unsatisfiable.
\end{proof}

\section{Proofs concerning $\Tot$}

\begin{proposition}\label{Tot is gentle}
    The theory $\Tot$ is gentle.
\end{proposition}

\begin{proof}
Let:
    $\vars(\varphi)$ equal $\{x_{1},\ldots,x_{n}\}$;
    $M_{i}$ be the maximum of $j$ such that $\ns^{j}(x_{i})$ shows up in $\varphi$;
    $M^{\prime}_{0}$ be the maximum of $j$ such that $\ns^{j}(a)$ appears in $\varphi$, and if it doesn't we set $M^{\prime}_{0}$ to $0$;
    $M_{0}=M^{\prime}_{0}+\sum_{i=1}^{n}(M_{i}+1)$;
    and take fresh variables $x_{i,j}$, for $0\leq i\leq n$ and $0\leq j\leq M_{i}$.
    We then flatten and Ackermannize $\varphi$, meaning:
    we replace any term $\ns^{j}(x_{i})$ by $x_{i,j}$, and any term $\ns^{j}(a)$ by $x_{0,j}$, in order to obtain the formula of equational logic $\varphi^{\prime}$;
    and define the formula $\varphi_{*}$ as $\varphi^{\prime}\wedge Fun(V)$, where $V=\{x_{i,j} : 0\leq i\leq n, 0\leq j\leq M_{i}\}$ and 
    \[Fun(V)=\bigwedge_{0\leq i,p\leq n}\bigwedge_{0\leq j< M_{i}}\bigwedge_{0\leq q<M_{p}}(x_{i,j}=x_{p,q})\rightarrow(x_{i,j+1}=x_{p,q+1}).\]
    Now, consider the set $\eq{V}$ of equivalence relations on $V$, which is easily computable; 
    $[x_{i,j}]_{E}$ shall denote in what follows the equivalence class under $E$ with representative $x_{i,j}$. 
    We define a subset $\eqp{V}$ of $\eq{V}$ such that $E$ is in $\eqp{V}$ if, and only if, when defining the interpretation of equational logic with domain $V/E$ and where $x_{i,j}$ is assigned the value $[x_{i,j}]_{E}$ (an interpretation we shall denote by $\int{V}{E}$), $\varphi_{*}$ is true in this interpretation (this can be decided algorithmically given the finiteness of $V/E$).

    For $E\in\eqp{V}$ we define a partial function $\ns_{E}$ on $V/E$ by making $\ns_{E}([x_{i,j}]_{E})=[x_{i,j+1}]_{E}$ for all $i\in\{0,\ldots,n\}$ and $j\in\{0,\ldots,M_{i}-1\}$ (notice that $\ns_{E}([x_{i,M_{i}}]_{E})$ may still be defined if $[x_{i,M_{i}}]_{E}=[x_{p,q}]_{E}$ for a $q\in\{0,\ldots,M_{p}-1\}$).
    This is well-defined: 
    if $[x_{i,j}]_{E}=[x_{p,q}]_{E}$ for $0\leq i,p\leq n$, $0\leq j<M_{i}$ and $0\leq q<M_{p}$, we have that $\int{V}{E}$ satisfies $x_{i,j}=x_{p,q}$;
    since it also satisfies $Fun(V)$, we have that it satisfies $x_{i,j+1}=x_{p,q+1}$, meaning that $\ns_{E}([x_{i,j}]_{E})=[x_{i,j+1}]_{E}=[x_{p,q+1}]_{E}=\ns_{E}([x_{p,q}]_{E})$.
    The partial function $\ns_{E}$ can be computed by an exhaustive search, as $V$, and thus $V/E$, is finite.

    We then let $B_{0}^{E}$ be the orbit of $[x_{0,0}]_{E}$ under $\ns_{E}$:
    for a partial function, this means either the list $\{\ns_{E}^{j}([x_{0,0}]_{E}) : j\in\mathbb{N}\}$ if $\ns_{E}$ is always defined on $\ns_{E}^{j}([x_{0,0}]_{E})$;
    or the list $\{\ns_{E}^{0}([x_{0,0}]_{E}), \ldots, \ns_{E}^{J}([x_{0,0}]_{E})\}$, if $\ns_{E}$ is defined on all $\ns_{E}^{j}([x_{0,0}]_{E})$ for $0\leq j\leq J-1$, but not on $\ns_{E}^{J}([x_{0,0}]_{E})$;
    this can be easily found algorithmically.\footnote{\label{footnote-in-proofs}The proof of \Cref{prop:tgpdec} continues from here.}
    Define $\eqpp{V}$ as the subset of $\eqp{V}$ where $2|B_{0}^{E}|\geq |V/E|$.
    For every $E\in\eqpp{V}$, define the interval $I(E)$ as
    \[I(E)=\begin{cases}
        [|V/E|,2|B_{0}^{E}|] & \text{if $\ns_{E}$ is defined for all of $B_{0}^{E}$,}\\
        \{n\in\No: n\geq|V/E|\}\cup\{\Inf\} & \text{otherwise,}
    \end{cases}\]
    and we state that 
    \[\spec(\varphi)=\bigcup_{E\in\eqpp{V}}I(E)\]
    if $\eqpp{V}$ is not empty, and $\spec(\varphi)=\emptyset$ otherwise.
    Given that the sets $I(E)$ are computable and either finite or cofinite, so is $\spec(\varphi)$ if the identity truly holds, meaning $\Tot$ is gentle.
    We prove the identity in three cases.
    \begin{enumerate}
    \item If $\eqpp{V}$ is not empty and $\ns_{E}$ is defined for all of $B_{0}^{E}$, for each $0\leq j\leq 2|B_{0}^{E}|-|V/E|$ take a set $B$ with cardinality $j$ disjoint from $V/E$, and we define a $\Tot$-interpretation $\A_{j}$ as follows.

    We make $\dom{\A_{j}}=(V/E)\cup B$, which then has $|V/E|+j\leq 2|B_{0}^{E}|$ elements.
    Of course $a^{\A_{j}}=[x_{0,0}]_{E}$.
    $\ns^{\A_{j}}(b)=\ns_{E}(b)$ for all $b$ where $\ns_{E}$ is defined, and $\ns^{\A_{j}}(b)=b$ otherwise:
    this way the orbit of $a^{\A_{j}}$ under $\ns^{\A}$ has $|B_{0}^{E}|$ elements, making $\A_{j}$ a $\Tot$-interpretation.
    And, finally, $x^{\A_{j}}=[x]_{E}$ for all variables $x\in V$, $x_{i}^{\A_{j}}=[x_{i,0}]_{E}$ for all variables $x_{i}$ in $\varphi$, and arbitrarily otherwise, so $\A_{j}$ satisfies $\varphi$ as $E\in\eqp{V}$.

    There cannot exist a $\Tot$-interpretation $\A$ that induces the equivalence $E$ on $V$ with fewer than $|V/E|$ elements, obviously.
    And there cannot exist a $\Tot$-interpretation $\A$ that induces the equivalence $E$ on $V$ with more than $2|B_{0}^{E}|$ elements as the orbit of $a^{\A}$ under $\ns^{\A}$ has necessarily $|B_{0}^{E}|$ elements.
    
    \item If $\eqpp{V}$ is not empty and $\ns_{E}$ is not defined over all of $B_{0}^{E}$, take any $j\in\mathbb{N}$ and a set $B=\{b_{1},\ldots,b_{j}\}$ with cardinality $j$ disjoint from $V/E$, and we define a $\Tot$-interpretation $\A_{j}$ as follows.

    First $\dom{\A_{j}}=(V/E)\cup B$, so $|\dom{\A_{j}}|=|V/E|+j$.
    Second, of course $a^{\A_{j}}=[x_{0,0}]_{E}$.
    Third:
    $\ns^{\A_{j}}(b)=\ns_{E}(b)$ for all $b$ where $\ns_{E}$ is defined;
    $\ns^{\A_{j}}(b)=b_{1}$ for the one element $b\in B_{0}^{E}$ where $\ns_{E}$ is not defined;
    $\ns^{\A_{j}}(b_{i})=b_{i+1}$ for $1\leq i\leq j-1$;
    and, for all elements $b$ where $\ns^{\A_{j}}$ hasn't been defined yet, including $b_{j}$, $\ns^{\A}(b)=b$.
    Notice that, this way, the orbit of $a^{\A_{j}}$ under $\ns^{\A_{j}}$ has size $|B_{0}^{E}|+j$, and since $2(|B_{0}^{E}|+j)\geq |V/E|+2j\geq |V/E|+j=|\dom{\A_{j}}|$ we get $\A_{j}$ is a $\Tot$-interpretation.
    And, finally, $x^{\A_{j}}=[x]_{E}$ for all variables $x\in V$, $x_{i}^{\A_{j}}=[x_{i,0}]_{E}$ for all variables $x_{i}$ in $\varphi$, and arbitrarily otherwise, so $\A_{j}$ satisfies $\varphi$ as $E\in\eqp{V}$.

    Of course an interpretation that satisfies $\varphi$ and induces the equivalence $E$ must have at least $|V/E|$ elements, so we are done.

    \item Suppose a $\Tot$-interpretation $\A$ satisfies $\varphi$: 
    we can change the values assigned to the variables $x_{i,j}$ while keeping $\varphi$ satisfied, as they are not in $\varphi$, so that $x_{i,j}^{\A}=(\ns^{\A})^{j}(x_{i}^{\A})$;
    take then the equivalence $E$ on $V$ such that $x_{i,j}Ex_{p,q}$ if $x_{i,j}^{\A}=x_{p,q}^{\A}$.
    Of course $\varphi_{*}$ is satisfied by $\int{V}{E}$, so $E\in\eqp{V}$.
    If $\ns_{E}$ is defined for all elements of the orbit of $[x_{0,0}]_{E}$ under $\ns_{E}$, we have that $|\dom{\A}|=2|B_{0}^{E}|$ and, since $|\dom{\A}|\geq |V/E|$, we get $2|B_{0}^{E}|\geq |V/E|$;
    if it is not, then it contains $M_{0}$ elements, and since $|V/E|\leq M_{0}$ we again get $2|B_{0}^{E}|\geq |V/E|$, proving that $\eqpp{V}$.
    Therefore, for $\varphi$ to be satisfiable we must have some $E\in\eqpp{V}$.
    
\end{enumerate}

\end{proof}

\begin{proposition}\label{prop:diffproof}
    The theory $\Tot$ is decidable.
\end{proposition}

\begin{proof}
    Follows from \Cref{Tot is gentle}: 
    a quantifier-free formula $\varphi$ is $\Tot$-satisfiable if and only if $\spec(\varphi)$ is not empty, something that is decidable.
\end{proof}

\begin{proposition}\label{prop:totSFW}
    The theory $\Tot$ is strongly finitely witnessable.
\end{proposition}

\begin{proof}
Let  $x_{1}$ through $x_{n}$ be the variables in a quantifier-free formula $\varphi$, $M_{i}$ be the maximum of $j$ such that $\ns^{j}(x_{i})$ occurs in $\varphi$, $M^{\prime}_{0}$ be the maximum of $j$ such that $\ns^{j}(x_{0})$ occurs in $\varphi$, $M_{0}$ be $M^{\prime}_{0}+\sum_{i=1}^{n}(M_{i}+1)$, and take fresh variables $x_{i,j}$, for $0\leq i\leq n$ and $0\leq j\leq M_{i}$.
    We state 
    \[\wit(\varphi)=\varphi\wedge\bigwedge_{i=0}^{n}\bigwedge_{j=0}^{M_{i}}x_{i,j}=\ns^{j}(x_{i})\]
    is a strong witness for $\Tot$.
    Of course it maps quantifier-free formulas into other quantifier-free formulas, and is computable.
    Furthermore, for $\overarrow{x}=\vars(\wit(\varphi))\setminus\vars(\varphi)$, it is obvious that $\Exists{\overarrow{x}}\wit(\varphi)$ implies $\varphi$, since $\wit(\varphi)$ itself already implies $\varphi$.
    Reciprocally, if the $\Tot$-interpretation $\A$ satisfies $\varphi$, we produce a new $\Tot$-interpretation $\B$ by changing the values assigned by $\A$ to those variables in $\overarrow{x}$ so that $x_{i,j}^{\B}=(\ns^{\A})^{j}(x_{i}^{\A})$;
    this way $\B$ satisfies $\wit(\varphi)$, and therefore $\A$ satisfies $\Exists{\overarrow{x}}\wit(\varphi)$.

    Now, take a finite set of variables $V$ (not to be confused with the $V$ used in the proof above that $\Tot$ is decidable), an arrangement $\delta_{V}$ over $V$, and a $\Tot$-interpretation $\A$ that satisfies $\wit(\varphi)\wedge\delta_{V}$: 
    there are two cases we consider;
    for simplicity, let $U$ denote $\vars(\wit(\varphi))$.
    \begin{enumerate}
        \item Suppose that the orbit of $a^{\A}$ under $\ns^{\A}$ is a subset of $U^{\A}\cup V^{\A}$, and we then define an interpretation $\B$ by making:
        $\dom{\B}=U^{\A}\cup V^{\A}$;
        $a^{\B}=a^{\A}$;
        $\ns^{\B}(b)=\ns^{\A}(b)$ whenever the latter value is in $\dom{\B}$, and otherwise $\ns^{\B}(b)=b$ (this way, the orbit of $a^{\B}$ under $\ns^{\B}$ is the same as the orbit of $a^{\A}$ under $\ns^{\A}$, and since $|\dom{\B}|\leq |\dom{\A}|$ we get $\B$ is a $\Tot$-interpretation);
        and $x^{\B}=x^{\A}$ for every variable $x\in U\cup V$, and arbitrarily otherwise (so $\dom{\B}=\vars(\wit(\varphi)\wedge\delta_{V})^{\B})$).

        It is clear that $\B$ satisfies $\delta_{V}$, given the definition of how it assigns values to variables;
        since $x_{i,j}^{\B}=x_{i,j}^{\A}$ and $x_{i,j}^{\A}=(\ns^{\A})^{j}(x_{i}^{\A})$, and thus $(\ns^{\B})^{j}(x_{i}^{\B})=(\ns^{\A})^{j}(x_{i}^{\A})$ (for $0\leq j\leq M_{i}$), we get $\B$ satisfies $\bigwedge_{i=0}^{n}\bigwedge_{j=0}^{M_{i}}x_{i,j}=\ns^{j}(x_{i})$.
        
        Finally, for any atomic subformula of $\varphi$ of the form $\ns^{j}(x_{i})=\ns^{q}(x_{p})$ (the cases $\ns^{j}(x_{i})=\ns^{q}(a)$ and $\ns^{j}(a)=\ns^{q}(a)$ being analogous), since $(\ns^{\B})^{j}(x_{i}^{\B})=(\ns^{\A})^{j}(x_{i}^{\A})$ and $(\ns^{\B})^{q}(x_{p}^{\B})=(\ns^{\A})^{q}(x_{p}^{\A})$, we get that the subformula is true in $\B$ if and only if it is true in $\A$;
        since $\varphi$ is quantifier-free, $\B$ satisfies $\varphi$, and thus $\wit(\varphi)$, and we are done.

        \item Suppose then that the orbit of $a^{\A}$ under $\ns^{\A}$ is not entirely contained in $U^{\A}\cup V^{\A}$:
        there are, however, at least $M_{0}+1$ (distinct) elements of it in $U^{\A}\cup V^{\A}$.
        For simplicity, assume $\{b_{1},\ldots, b_{R}\}$ is an enumeration of $V^{\A}\setminus U^{\A}$.
        We then define an interpretation $\B$ as follows:
        $\dom{\B}=U^{\A}\cup V^{\A}$;
        $a^{\B}=a^{\A}$;
        $\ns^{\B}(c)=\ns^{\A}(c)$ whenever the latter is in $U^{\A}$, $(\ns^{\B})^{M_{0}+1}(a^{\B})=b_{1}$, $\ns^{\B}(b_{r})=b_{r+1}$ for $1\leq r<R$, and if $\ns^{\B}(c)$ hasn't been defined yet we make it equal to $c$ (the orbit of $a^{\B}$ under $\ns^{\B}$ has then $M_{0}+R+1$ elements, and since $|U^{\A}|\leq 2M_{0}+1$ and $|V^{\A}\setminus U^{\A}|=R$ we conclude $\B$ is a $\Tot$-interpretation);
        and $x^{\B}=x^{\A}$ for every variable $x\in U\cup V$, and arbitrarily otherwise (so $\dom{\B}=\vars(\wit(\varphi)\wedge\delta_{V})^{\B})$).

        We finally prove $\B$ satisfies $\wit(\varphi)\wedge\delta_{V}$ just as in the item above.
\end{enumerate}
\end{proof}



\begin{proposition}
    The theory $\Tot$ has the finite model property.
\end{proposition}

\begin{proof}
    Follows from \Cref{prop:totSFW} and the fact that strong finite witnessability implies the finite model property, as shown in \cite{FroCoS}.
\end{proof}

\begin{proposition}
    The theory $\Tot$ has a computable minimal model function.
\end{proposition}

\begin{proof}
    Follows from \Cref{lem:SFWandDECimpCMMF,prop:diffproof,prop:totSFW}.
\end{proof}

\begin{proposition}
    The theory $\TMn\oplus\Tot$ is undecidable.
\end{proposition}

\begin{proof}
    Assume instead $\TMn\oplus\Tot$ is decidable, and let us define a function $\NFT:\No\rightarrow\{0,1\}$ and formulas $\varphi_{n}$ by making:
    $\NFT(1)=\NFT(3)=1$ and $\NFT(2)=\NFT(4)=0$;
    assuming $\NFT$ defined up to $2n$, $\varphi_{n}$ equal to $\orb{n+1}(a)\wedge\varphi_{\geq \NFT_{1}(2n)+2}^{\eqs}$ (where, as before, $\NFT_{1}(n)=|\{1\leq i\leq n : \NFT(i)=1\}|$), and
        \[\NFT(2n+1)=\NFT(2n+2)=\begin{cases}
        1 & \text{if $\varphi_{n}$ is $\TMn\oplus\Tot$-satisfiable},\\
        $0$ & \text{otherwise};\\
    \end{cases}\]
    Of course $\NFT$ is computable, but we shall show that $\NFT=\nft$ and reach a contradiction; 
    assume this is true for all values up to $2n$ for $n\geq 2$, meaning in particular that $\NFT_{1}(n)=\nft_{1}(n)$, and we shall show $\nft(2n+1)=\nft(2n+2)=1$ if and only if $\varphi_{n}$ is $\TMn\oplus\Tot$-satisfiable.

    If it is satisfiable, there is a $\TMn\oplus\Tot$-interpretation $\A$ that satisfies $\varphi_{n}$, and thus $\orb{n+1}(a)$ as well as $\varphi_{\geq\nft_{1}(2n)+2}^{\eqs}$;
    from the axiom $\orb{n+1}(a)\rightarrow\psi_{\leq 2n+2}$ of $\Tot$ we get $\A$ has at most $2n+2$ elements, and from the axiomatization of $\TMn$ we get that it has at most $\nft_{1}(2n+2)$ elements satisfying  $\os(x)=x$.
    The fact that $\A$ satisfies $\varphi_{\geq\nft_{1}(2n)+2}^{\eqs}$ implies it has at least $\nft_{1}(2n)+2$ elements satisfying  $\os(x)=x$, and these two last facts are only compatible if $\nft(2n+1)=\nft(2n+2)=1$.

    Reciprocally, suppose $\nft(2n+1)=\nft(2n+2)=1$, and so there exists a $\TMn$-interpretation $\A$ with exactly $\nft_{1}(2n+2)=\nft_{1}(2n)+2$ elements satisfying  $\os(x)=x$, and $\nft_{0}(2n+2)=\nft_{0}(2n)$ satisfying  $\os(x)\neq x$ (and thus $2n+2$ elements in total, which we name $a_{1}$ through $a_{2n+2}$).
    Extend $\A$ to a $\Sat\oplus\Ss$-interpretation $\B$ by making $a^{\B}=a_{1}$, $\ns^{\B}(a_{i})=a_{i+1}$ for $1\leq i\leq n$, $\ns^{\B}(a_{n+1})=a_{n+1}$, and $\ns^{\B}(a_{i})=a_{i}$ for $n+2\leq i\leq 2n+2$:
    $\orb{n+1}(a)$ is then true in $\B$, while all $\orb{m}(a)$, for $m\neq n+1$, are obviously false. 
    We prove that $\B$ is then a $\TMn\oplus\Tot$-interpretation that satisfies $\varphi_{n}$, which shall finish the proof.
    $\B$ is certainly a $\TMn$-interpretation, vacuously satisfies all axioms $\orb{m}(a)\rightarrow\psi_{\leq 2m}$ for $m\neq n+1$, and satisfies $\orb{n+1}(a)\rightarrow\psi_{\leq 2n+2}$ given that it satisfies $\orb{n+1}(a)$ and has $2n+2$ elements, making of it a $\TMn\oplus\Tot$-interpretation.
    Furthermore, as mentioned before it satisfies $\orb{n+1}(a)$, and satisfies $\varphi_{\geq\nft_{1}(2n)+2}^{\eqs}$ since it has $\nft_{1}(2n+2)=\nft_{1}(2n)+2$ elements that satisfy  $\os(x)=x$.
\end{proof}

\section{Proofs concerning $\Tgp$}

\begin{proposition}\label{prop:tgpdec}
    The theory $\Tgp$ is decidable.
\end{proposition}

\begin{proof} 
This is essentially the same as the proof of \Cref{prop:diffproof}, the proofs being exactly the same up to \Cref{footnote-in-proofs}.
        The difference is that $\varphi$ is $\Tgp$-satisfiable if and only if there is an $E\in\eqp{V}$ such that
        \[|B_{0}|+F(|B_{0}|)\geq |V/E|.\]
        Now, we may not be able to calculate $F(|B_{0}|)$, but we can equivalently write this condition as $F(|B_{0}|)\geq |V/E|-|B_{0}|$, and such tests being computable are a prerequisite for $F$.
        \end{proof}

        \begin{proposition} 
        The theory $\Tinfty\oplus\Tgp$ is undecidable. 
        \end{proposition}
        
        \begin{proof} 
        Consider the formulas $\orb{n}$:
        we state that they are $\Tinfty\oplus\Tgp$-satisfiable if and only if $F(n)=\Inf$, what we know cannot be tested algorithmically.
        
        Take first an $n$ such that $F(n)=\Inf$, and define an interpretation $\A$ as follows:
        $\dom{\A}=\mathbb{N}$ (so $\A$ is a $\Tinfty$-interpretation) and $a^{\A}=0$;
        $\ns^{\A}(i)=i+1$ for all $i\neq n-1$, and $\ns^{\A}(n-1)=0$, so the orbit of $0$ is $\{0,\ldots,n-1\}$, meaning $\A$ satisfies $\orb{n}$ and, since $F(n)=\Inf$ and $\A$ is infinite we have that it is a $\Tgp$-interpretation;
        and $x^{\A}$ can be defined arbitrarily for all variables $x$.
        This means $\A$ is a $\Tinfty\oplus\Tgp$-interpretation that satisfies $\varphi$.

        Reciprocally, suppose $\A$ is a $\Tinfty\oplus\Tgp$-interpretation that satisfies $\orb{n}$: 
        if $F(n)\in\No$ we have $|\dom{\A}|\leq F(n)+n$, which is finite and thus contradicts the fact that $\A$ is supposed to be a $\Tinfty$-interpretation. 
        Thus $F(n)=\Inf$.
        \end{proof}

\section{\tp{Proof of \Cref{gentle-recovery}}{Proof of Theorem~\ref{gentle-recovery}}}

The key to the proof is the following result due to Fontaine.

\begin{lemma}[{\cite[Corollary~1]{gentle}}] \label{fontaine-lemma}
    Let $\T_1$ and $\T_2$ be theories over disjoint signatures $\Sigma_1$ and $\Sigma_2$, respectively. Suppose that it is decidable whether $\spec(\T_1,\varphi_1) \cap \spec(\T_2,\varphi_2) = \emptyset$, where $\varphi_1$ and $\varphi_2$ are conjunctions of literals over the signatures $\Sigma_1$ and $\Sigma_2$, respectively. Then, $\T_1 \oplus \T_2$ is decidable.
\end{lemma}

In light of the lemma, the following implies \Cref{gentle-recovery}.

\begin{lemma}
    Let $\T_1$ and $\T_2$ be decidable theories over disjoint signatures $\Sigma_1$ and $\Sigma_2$ respectively. Suppose that $\T_1$ is gentle and $\T_2$ has computable finite spectra. Then, it is decidable whether $\spec(\T_1,\varphi_1) \cap \spec(\T_2,\varphi_2) = \emptyset$, where $\varphi_1$ and $\varphi_2$ are conjunctions of literals over the signatures $\Sigma_1$ and $\Sigma_2$ respectively.
\end{lemma}
\begin{proof}
    Let $\varphi_1$ and $\varphi_2$ be conjunctions of literals over the signatures $\Sigma_1$ and $\Sigma_2$ respectively. We describe our decision procedure as follows. Since $\T_1$ is gentle, $\spec(\T_1,\varphi_1)$ is either of the form $S$ or $S \cup \{n \in \N \mid n \ge k\}$ for some $k \in \No$, where $S \subset \No$ is a finite set. We have $S \cap \spec(\T_2,\varphi_2) = \emptyset$ if and only if $n \notin \spec(\T_2,\varphi_2)$ for each $n \in S$, which we can check algorithmically since $\T_2$ has computable finite spectra. We also have $\{n \in \N \mid n \ge k\} \cap \spec(\T_2,\varphi_2) = \emptyset$ if and only if $\varphi_2 \land \neq(x_{1},\ldots,x_{k})$ is $\T_2$-unsatisfiable (where the variables $x_i$ are fresh). These computations allow us to determine whether $\spec(\T_1,\varphi_1) \cap \spec(\T_2,\varphi_2) = \emptyset$.
\end{proof}

We also prove here that \Cref{gentle-recovery} is a strengthening of \Cref{gentle-recovery-fontaine}.

\begin{proposition}
    If a theory $\T$ is gentle, then $\T$ has computable finite spectra.
\end{proposition}
\begin{proof}
    Let $\varphi$ be a quantifier-free formula, and let $k \in \No$. If $\T$ is gentle, then we can compute an explicit representation of the set $\spec(\T,\varphi)$, from which we can decide whether $k \in \spec(\T,\varphi)$.
\end{proof}

\begin{proposition}
    If a theory $\T$ is finitely axiomatizable, then $\T$ has computable finite spectra.
\end{proposition}
\begin{proof}
    Let $\varphi$ be a quantifier-free formula, and let $k \in \No$. Let $\Sigma$ be the signature over which $\T$ is defined. We may assume that $\Sigma$ only contains the symbols appearing in $\ax(\T) \cup \{\varphi\}$ so that, in particular, $\Sigma$ is finite. We can enumerate the $\Sigma$-interpretations of size $k$, checking whether any of them satisfy all of the formulas in $\ax(\T) \cup \{\varphi\}$. If we find a $\T$-interpretation of size $k$ satisfying $\varphi$, then $k \in \spec(\T,\varphi)$; otherwise, $k \notin \spec(\T,\varphi)$.
\end{proof}

\begin{proposition}
    Suppose that for a theory $\T$, there is an algorithm that, given a conjunction $\varphi$ of literals, outputs a finite set $S \subset \N$ such that $\spec(\T,\varphi) = S$. Then, $\T$ has computable finite spectra.
\end{proposition}
\begin{proof}
    Let $\varphi$ be a quantifier-free formula, and let $k \in \No$. Without loss of generality, $\varphi$ is a conjunction of literals (the general case follows by putting $\varphi$ in disjunctive normal form). We can compute a finite set $S \subset \N$ such that $\spec(\T,\varphi) = S$, from which we can decide whether $k \in \spec(\T,\varphi)$.
\end{proof}

\section{\tp{Proofs concerning \Cref{gentle-recovery-ex}}{Proof of Result \ref{gentle-recovery-ex}}}

\begin{proposition}
    For any $n \in \No$, the theory $\T_{\leq n}$ is gentle (and therefore decidable).
\end{proposition}
\begin{proof}
    Let $\varphi$ be a conjunction of literals. If $\varphi$ is unsatisfiable in equational logic, then $\spec(\T_{\leq n},\varphi) = \emptyset$. Otherwise, let $m$ be the size of the size of the smallest interpretation that satisfies $\varphi$. Then, $\spec(\T_{\leq n},\varphi) = [m,n]$.
\end{proof}

That $\Tgr$ is decidable is proved in \Cref{Tgr-dec}.

\begin{proposition}
    The theory $\Tgr$ has computable finite spectra.
\end{proposition}
\begin{proof}
    Let $\varphi$ be a conjunction of literals, and let $k \in \No$. Write $\varphi = \varphi_1 \land \varphi_2$, where $\varphi_1$ contains the equalities and disequalities in $\varphi$ and $\varphi_2$ contains the literals of the form $P_n$ and $\lnot P_n$ in $\varphi$.

    We describe our decision procedure as follows. If $\varphi_1$ is unsatisfiable in equational logic, then $\varphi$ is $\Tgr$-unsatisfiable, so $k \notin \spec(\Tgr,\varphi)$. Otherwise, let $m$ be the size of the smallest interpretation that satisfies $\varphi_1$. If $k < m$, then $k \notin \spec(\Tgr,\varphi)$. So assume that $k \ge m$.

    We claim that, in this case, $k \in \spec(\Tgr,\varphi)$ if and only if for every $n$ such that the literal $P_n$ is in $\varphi_2$, we have $F(n) \ge k$. This is because if the latter condition holds, we can extend an interpretation of size $k$ that satisfies $\varphi_1$ to a $\Tgr$-interpretation $\A$ of size $k$ satisfying $\varphi$ by setting $P_n^\A$ to true for every $n$ such that the literal $P_n$ is in $\varphi_2$ and setting $P_n^\A$ to false otherwise. Otherwise, there is some $n$ such that the literal $P_n$ is in $\varphi_2$ and $F(n) < k$. Then, there is no interpretation of $\varphi_2$ of size $k$ (since $P_n \rightarrow \psi_{\leq F(n)}$ is an axiom of $\Tgr$), so $k \notin \spec(\Tgr,\varphi)$.
\end{proof}

\begin{proposition}
    The theory $\Tgr$ is not gentle.
\end{proposition}
\begin{proof}
    It suffices to show that it is undecidable whether $\spec(\Tgr,\varphi)$ is infinite. And indeed, $\spec(\Tgr,P_n)$ is infinite if and only if $F(n) = \aleph_0$, which is undecidable by our assumptions on $F$.
\end{proof}

\begin{proposition}
    The theory $\Tgr$ is not finitely axiomatizable.
\end{proposition}
\begin{proof}
    If $\T$ is a finitely axiomatizable $\Spn$-theory, then the predicate $P_n$ does not appear in the axiomatization of $\T$ for some $n$ such that $F(n) < \aleph_0$. Then, $P_n \rightarrow \psi_{\leq F(n)}$ is $\Tgr$-valid but not $\T$-valid.
\end{proof}

\begin{proposition}
    There is no algorithm that, given a conjunction of literals $\varphi$ in the language of $\Tgr$, outputs a finite set $S \subset \N$ such that $\spec(\Tgr,\varphi) = S$
\end{proposition}
\begin{proof}
    Indeed, $\spec(\Tgr,\top)$ is not even a finite set.
\end{proof}

\begin{proposition}
    For any $n \in \No$, the theory $\T_{\leq n}$ is not stably infinite (and therefore neither strongly polite nor shiny).
\end{proposition}
\begin{proof}
    The theory has no models of size greater than $n$ and hence no infinite models.
\end{proof}

\begin{proposition}
    The theory $\Tgr$ is not stably infinite (and therefore neither strongly polite nor shiny).
\end{proposition}
\begin{proof}
    Let $n \in \No$ be such that $F(n) < \aleph_0$. Then, $P_n$ has no infinite $\Tgr$-models, since $P_n \rightarrow \psi_{\leq F(n)}$ is an axiom of $\Tgr$.
\end{proof}

\section{\tp{Proof of \Cref{new-method}}{Proof of Result \ref{new-method}}}

In light of \Cref{fontaine-lemma}, the following implies \Cref{new-method}.

\begin{lemma}
    Let $\T_1$ and $\T_2$ be decidable theories over disjoint signatures $\Sigma_1$ and $\Sigma_2$ respectively. Suppose that $\T_1$ is smooth and has a computable minimal model function and that $\T_2$ is infinitely decidable. Then, it is decidable whether $\spec(\T_1,\varphi_1) \cap \spec(\T_2,\varphi_2) = \emptyset$, where $\varphi_1$ and $\varphi_2$ are conjunctions of literals over the signatures $\Sigma_1$ and $\Sigma_2$ respectively.
\end{lemma}
\begin{proof}
    Let $\varphi_1$ and $\varphi_2$ be conjunctions of literals over the signatures $\Sigma_1$ and $\Sigma_2$ respectively. We describe our decision procedure as follows. If $\varphi_1$ is $\T_1$-unsatisfiable, then $\spec(\T_1,\varphi_1) \cap \spec(\T_2,\varphi_2) = \emptyset$, so suppose $\varphi_1$ is $\T_1$-satisfiable.

    If $\minmod_{\T_1}(\varphi_1) = n$ for some $n \in \No$, then $\spec(\T_1,\varphi_1) = \{m \in \No \mid m \ge n\} \cup \{\aleph_0\}$, since $\T_1$ is smooth. In this case, we have $\spec(\T_1,\varphi_1) \cap \spec(\T_2,\varphi_2) = \emptyset$ if and only if $\varphi_2$ has no $\T_2$-interpretations of size at least $n$, which happens exactly when $\varphi_2 \land \neq(x_{1},\ldots,x_{n})$ is $\T_2$-unsatisfiable (where the variables $x_i$ are fresh).

    If $\minmod_{\T_1}(\varphi_1) = \aleph_0$, then $\spec(\T_1,\varphi_1) = \{\aleph_0\}$. In this case, we have $\spec(\T_1,\varphi_1) \cap \spec(\T_2,\varphi_2) = \emptyset$ if and only if $\varphi_2$ is not satisfied by any infinite $\T_2$-interpretation, which we can check since $\T_2$ is infinitely decidable.
\end{proof}

\section{\tp{Proofs concerning \Cref{theo:onlybencomb}}{Proof of Result \ref{theo:onlybencomb}}}

\begin{proposition}\label{prop:tbSI}
    The theory $\Tb$ is not stably-infinite.
\end{proposition}

\begin{proof}
    $P_{1}$ is only satisfied by a $\Tb$-interpretation with one element.
\end{proof}

\begin{proposition}
    The theory $\Tb$ is not strongly polite.
\end{proposition}

\begin{proof}
    Follows from \Cref{prop:tbSI}, which implies $\Tb$ is not smooth.
\end{proof}

\begin{proposition}
    The theory $\Tinfty$ is not strongly polite.
\end{proposition}

\begin{proof}
    $\Tinfty$ was proven in \cite{CADE} not to be strongly finitely witnessable, so the result follows.
\end{proof}

\begin{proposition}
    The theory $\Tb$ is not shiny.
\end{proposition}

\begin{proof}
    Follows from \Cref{prop:tbSI}, which implies $\Tb$ is not smooth.
\end{proof}

\begin{proposition}
    The theory $\Tinfty$ is not shiny.
\end{proposition}

\begin{proof}
    $\Tinfty$ does not have the finite model property, and therefore is not shiny, since it is not trivial but has no finite models.
\end{proof}

\begin{proposition}
    The theory $\Tb$ does not have computable finite spectra.
\end{proposition}
\begin{proof}
    For all $n \ge 2$, we have $1 \in \spec(\Tb,P_n)$ if and only if $h(n)=0$, which is undecidable by assumption.
\end{proof}

\begin{proposition}
    The theory $\Tinfty$ has a computable minimal model function.
\end{proposition}

\begin{proof}
    Proven in \cite{LPAR}.
\end{proof}

\begin{proposition}
    The theory $\Tinfty$ is smooth.
\end{proposition}

\begin{proof}
    Obvious, but proven in \cite{CADE}.
\end{proof}

\begin{proposition}\label{prop:Tbdec}
    The theory $\Tb$ is decidable.
\end{proposition}

\begin{proof}
    Given a quantifier-free formula $\varphi$, assume without loss of generality $\varphi$ is a cube and write $\varphi=\varphi_{1}\wedge\varphi_{2}$, where $\varphi_{1}$ contains only equalities and disequalities, and $\varphi_{2}$ contains only the predicates $P_{n}$ and their negations.
    We state $\varphi$ is $\Tb$-satisfiable if, and only if, $\varphi_{1}$ is satisfiable in equational logic, $\varphi_{2}$ does not contain a predicate and its negation, and if $\varphi_{2}$ contains the literal $P_{1}$, then $\varphi_{1}$ has a model in equational logic of size $1$ and $\varphi_{2}$ contains no other positive literals.

    If $\varphi$ is $\Tb$-satisfiable so is $\varphi_{1}$, and by forgetting the predicates on a $\Tb$-interpretation that satisfies $\varphi_{1}$ we get an interpretation in equational logic that satisfies that formula.
    Of course $\varphi_{2}$ is also satisfiable, and then we cannot have both $P_{n}$ and $\neg P_{n}$ in $\varphi_{2}$.
    Finally, if $\varphi_{2}$ contains the literal $P_{1}$, any $\Tb$-interpretation $\A$ that satisfies $\varphi$ must satisfy $P_{1}$ as well, and by the axiom $P_{1}\rightarrow\psi_{=1}$ of $\Tb$ we get $\A$ has only one element;
    by forgetting its predicates, we get an interpretation in equational logic that satisfies $\varphi_{1}$ and has only one element.
    Furthermore, by the axioms $P_{1}\rightarrow\neg P_{n}$ we get $\A$ satisfies $\neg P_{n}$ for $n\geq 2$, and thus $\varphi_{2}$ cannot contain the positive literals $P_{n}$.

    For the reciprocal, start by assuming $\varphi_{2}$ does not contain the literal $P_{1}$, and since $\varphi_{1}$ is satisfiable in equational logic it has an infinite model $\A$, as equational logic is stably-infinite;
    turn $\A$ into a $\Spn$-interpretation by setting $P_{n}^{\A}$ to true if the literal $P_{n}$ occurs in $\varphi_{2}$, and otherwise to false, meaning it satisfies $\varphi_{2}$ as well (this is possible because $\varphi_{2}$ cannot contain both a literal and its negation).
    $\A$ then satisfies the axioms $P_{1}\rightarrow\psi_{=1}$ and $P_{1}\rightarrow \neg P_{n}$ vacuously, as it does not satisfy $P_{1}$;
    and it satisfies the axioms $P_{n}\rightarrow\psi_{\geq m}$ for $\nf(n)=1$ since it is infinite and therefore satisfies all $\psi_{\geq m}$.
    If $\varphi_{2}$ contains $P_{1}$, we take an interpretation $\A$ in equational logic that satisfies $\varphi_{1}$ with only one element, and make it into a $\Spn$-interpretation by setting $P_{1}^{\A}$ to true, and all other $P_{n}^{\A}$ to false (so $\A$ satisfies $\varphi_{2}$ as well, since that formula cannot contain any positive literals other than $P_{1}$).
    It satisfies $P_{1}\rightarrow\psi_{=1}$ as it contains only one element, $P_{1}\rightarrow\neg P_{n}$ as all $P_{n}$ different from $P_{1}$ are false in $\A$, and $P_{n}\rightarrow\psi_{\geq m}$ (for $n\geq 2$ with $\nf(n)=1$) vacuously.
\end{proof}

\begin{proposition}
    The theory $\Tb$ is infinitely decidable.
\end{proposition}

\begin{proof}
    Given a quantifier-free formula $\varphi$, we assume, without loss of generality, that $\varphi$ is a cube;
    write then $\varphi=\varphi_{1}\wedge\varphi_{2}$, where $\varphi_{1}$ contains only equalities and disequalities, and $\varphi_{2}$ contains only the predicates $P_{n}$ and their negations.
    As proven in \Cref{prop:Tbdec}, $\Tb$ is decidable, so we only need to worry whether $\Tb$-satisfiable $\varphi$ are satisfied by an infinite $\Tb$-interpretation:
    we state that such $\varphi$ have infinite $\Tb$-models if and only if $P_{1}$ does not appear as a literal in $\varphi_{2}$.

    One direction is obvious:
    if $\varphi_{1}$ contains $P_{1}$ and the $\Tb$-interpretation $\A$ satisfies $\varphi$, by the axiom $P_{1}\rightarrow \psi_{=1}$ we get $\A$ can only have one element.
    If $P_{1}$ is not in $\varphi_{2}$ we proceed as in the proof of \Cref{prop:Tbdec}:
    as equational logic is stably-infinite, take an infinite interpretation $\A$ that satisfies $\varphi_{1}$;
    and set $P_{n}^{\A}$ to true if and only if the positive literal $P_{n}$ occurs in $\varphi_{2}$, so $\A$, as a $\Spn$-interpretation, satisfies $\varphi_{2}$ as well.
\end{proof}

\end{document}